%% file: quantumMain.tex
\documentclass[a4paper,twocolumn,11pt,accepted=2021-06-22]{quantumarticle}
\pdfoutput=1
\usepackage{graphicx}
\usepackage{float}
\usepackage{dcolumn}
\usepackage{bm}
\usepackage[table]{xcolor}
\usepackage{booktabs}
\usepackage{physics}
\usepackage{optidef}
\usepackage{amsthm}
\usepackage{amsmath}
\usepackage{amssymb}
\usepackage[breaklinks]{hyperref}

\usepackage[normalem]{ulem}
\usepackage{comment}
\usepackage{tcolorbox}
\usepackage[framed,autolinebreaks,useliterate]{mcode}
\usepackage{comment}
\newtheorem{thm}{Theorem}
\newtheorem{lem}{Lemma}
\newtheorem{prop}{Observation}

\definecolor{darkgreen}{rgb}{0,0.5,0}

\newcommand{\vicky}[1]{\textcolor{blue}{[vicky: #1]}}

\newcommand{\out}[1]{\textcolor{darkgreen}{\sout{#1}}} 
\newcommand{\inn}[1]{\textcolor{blue}{#1}}

\newcommand{\cH}{\mathcal{H}}
\newcommand{\cQ}{\mathcal{Q}}
\newcommand{\cB}{\mathcal{B}}
\newcommand{\cC}{\mathcal{C}}
\newcommand{\cO}{\mathcal{O}}
\newcommand{\cU}{\mathcal{U}}
\newcommand{\pp}{\mathbf{p}}

\begin{document}

\title{Characterising and bounding the set of quantum behaviours in contextuality scenarios}

\author{Anubhav Chaturvedi}
\email{anubhav.chaturvedi@research.iiit.ac.in}
\affiliation{Institute of Theoretical Physics and Astrophysics, National Quantum Information Centre, Faculty of Mathematics, Physics and Informatics, University of Gdansk, 80-952 Gdansk, Poland}
\affiliation{International Centre for Theory of Quantum Technologies, University of Gdansk, 80-308, Gdansk, Poland}

\author{M\'at\'e Farkas}
\affiliation{ICFO-Institut de Ciencies Fotoniques, The Barcelona Institute of Science and Technology,  08860 Castelldefels, Spain}

\author{Victoria J Wright}
\affiliation{International Centre for Theory of Quantum Technologies, University of Gdansk, 80-308, Gdansk, Poland}
\affiliation{ICFO-Institut de Ciencies Fotoniques, The Barcelona Institute of Science and Technology,  08860 Castelldefels, Spain}

\begin{abstract}

       The predictions of quantum theory resist generalised noncontextual explanations. In addition to the foundational relevance of this fact, the particular extent to which quantum theory violates noncontextuality limits available quantum advantage in communication and information processing. In the first part of this work, we formally define contextuality scenarios via prepare-and-measure experiments, along with the polytope of general contextual behaviours containing the set of quantum contextual behaviours. {This framework allows us to recover several properties of set of quantum behaviours in these scenarios, including contextuality scenarios and associated noncontextuality inequalities that require for their violation the individual quantum preparation and measurement procedures to be mixed states and unsharp measurements}. With the framework in place, we formulate novel semidefinite programming relaxations for bounding these sets of quantum contextual behaviours. Most significantly, to circumvent the inadequacy of pure states and projective measurements in contextuality scenarios, we present a novel unitary operator based semidefinite relaxation technique. We demonstrate the efficacy of these relaxations by obtaining tight upper bounds on the quantum violation of several noncontextuality inequalities and identifying novel maximally contextual quantum strategies. To further illustrate the versatility of these relaxations, we demonstrate \textit{monogamy of preparation contextuality} in a tripartite setting, and present a secure semi-device independent quantum key distribution scheme powered by quantum advantage in parity oblivious random access codes. 

\end{abstract}
\maketitle

\input{Introduction}

\input{ContextualityScenarios}

\input{sdprelax}

\input{TightBounds}

\input{Applications}

\input{Summary}

\section*{Note added}
{While finalising this article, we became aware of the related work in Ref.~\cite{tavakoli2020bounding}, which also features a hierarchy of SDP relaxations for bounding the set of quantum contextual behaviours. The hierarchy in Ref.~\cite{tavakoli2020bounding} is considerably different from the one presented here. {The authors} employ a single unhinged moment matrix, and to accommodate the inadequacy of pure states and projective measurements, they utilise additional localising moment matrices, each hinged on a positive semidefinite operator. For contextuality scenarios with {only trivial} measurement equivalences, their hierarchy should, to the best of our knowledge, perform at least as well as the one presented here. However, in the test cases we considered, {the present hierarchy yielded tight bounds with a smaller number of total variables.} 
 {Furthermore, }in general contextuality scenarios with non-trivial operational equivalences {our approach employing unitary operators } provides faster convergence. On the other hand, the hierarchy in Ref.~\cite{tavakoli2020bounding} admits an interesting generalisation to quantifying the simulation cost of contextuality.}

\section{Acknowledgements} 
The authors thank Piotr Mironowicz for the first co-implementation of the relaxation, and Miguel Navascu\'es for the idea in Lemma \ref{unitarylemma}. We thank Debashis Saha, Nikolai Miklin, Ana Bel\'en Sainz and John H.~Selby for fruitful discussions. AC acknowledges financial support by FNP grants First TEAM/2016-1/5,and NCN grant SHENG 2018/30/Q/ST2/00625. VJW acknowledges support from the Foundation for Polish Science (IRAP project, ICTQT, contract no.~2018/MAB/5, co-financed by EU within Smart Growth Operational Programme). MF and VJW acknowledge funding from the Government of Spain (FIS2020-TRANQI and Severo Ochoa CEX2019-000910-S), Fundaci\'o Cellex, Fundaci\'o Mir-Puig, Generalitat de Catalunya (CERCA, AGAUR SGR 1381 and QuantumCAT). MF was supported by the ERC AdG CERQUTE. The numerical optimization was carried out using \href{https://ncpol2sdpa.readthedocs.io/en/stable/index.html}{Ncpol2sdpa} \cite{wittek2015algorithm}, \href{https://yalmip.github.io/}{YALMIP} \cite{Lofberg2004} and \href{https://www.mosek.com/documentation/}{SDPT3}~\cite{toh1999sdpt3}. AC acknowledges \href{https://youtu.be/MM62wjLrgmA}{Schism by Tool} for creative support. 

\onecolumn

\bibliographystyle{alphaurl}
\bibliography{cite}

\input{Appendix}

\end{document}

%% file: Introduction.tex
\section{Introduction} 
The Leibnizian methodological principle of ``ontological identity of empirical indiscernibles" \cite{von1956leibniz} creates a bridge across the schism dividing the ``empiricist", and ``realist" viewpoints on physics \cite{spekkens2019ontological}.
Generalised noncontextuality \cite{spekkens2005contextuality} embodies this principle, and serves to characterise operational physical theories that allow for simultaneously realist and Leibnizian explanations.

Although born of metaphysical considerations, noncontextuality has been shown to be equivalent to a natural operational notion of classicality, namely, simplex embeddability of convex general stochastic operational theories \cite{schmid2019characterization,schmid2020structure}. 
It is then not surprising that the quantum \emph{violation of noncontextuality}\footnote{{Quantum theory is said to violate noncontextuality as shorthand for admitting no simultaneously Leibnizian and realist explanation. More explicitly, it cannot be explained by a noncontextual ontological (hidden variable) model \cite{spekkens2005contextuality,spekkens2019ontological,schmid2020unscrambling}}} fuels quantum advantage in broad classes of information processing, and cryptographically significant communication tasks \cite{ExpContext,PhysRevLett.102.010401,chailloux2016optimal,PhysRevLett.102.010401,ghorai2018optimal,Noob1,PhysRevA.100.022108,saha2019state,schmid2018contextual,PhysRevA.100.042116,emeriau2020quantum,yadavalli2020contextuality}.

There is, therefore, both foundational and technological motivation for characterising the peculiar nature of contextual quantum behaviour, and quantifying the extent of contextuality in terms of quantum violation of \emph{noncontextuality inequalities}. In analogy with quantum non-locality \cite{nonlocality}, quantum theory surpasses these noncontextual limitations, but not necessarily to the maximum mathematically possible extent \cite{banik2015limited}. {Bounding this violation provides foundational insights, such as establishing how distinct operationally indistinguishable entities must be in order to explain quantum predictions \cite{Noob70747,marvian2020inaccessible}}, and finds technological applications in capping quantum advantage. Despite this motivation, the nature and the extent of available quantum contextuality remains largely uncharacterised.
\begin{figure}[ht]
    \centering
    \includegraphics[width=\linewidth]{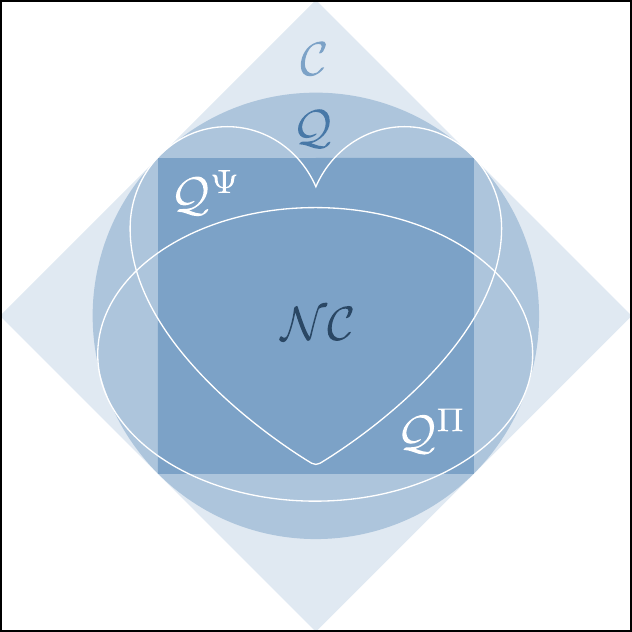} 
    \caption{\label{awesome} \textit{The landscape of contextual behaviours:  } The diagram illustrates the relationship between sets of behaviours achievable in different theories in a generic contextuality scenario. In direct analogy with the polytopes of no-signalling and local-causal correlations in Bell experiments, the set of all valid contextual behaviours, $\mathcal{C}$, and the set of noncontextual behaviours, $\mathcal{NC}$, form polytopes. Continuing the analogy, the set of quantum behaviours, $\mathcal{Q}$, although convex, does not form a polytope, and satisfies $\mathcal{NC}\subsetneq\cQ\subsetneq\cC$ in general. We find contextuality scenarios in which the subsets of behaviours obtained from pure quantum states $\mathcal{Q}^\Psi$ (possibly convex), and projective measurements $\mathcal{Q}^{\Pi}$ (convex) form strict subsets of the set of quantum contextual behaviours (see Observations~\ref{QNotEqualsQPi} and~\ref{QNotEqualsQPsi}), thereby demonstrating that in such scenarios mixed quantum states and unsharp measurements can access greater contextuality.}
\end{figure}

In an effort to address this shortcoming, in the first part of this work, we explore and elucidate the rich landscape of contextual behaviours by formally defining \textit{contextuality scenarios} via prepare-and-measure experiments. This formalism allows us to define the sets of noncontextual, quantum and contextual behaviours. We can then prove some basic facts about the quantum set, summarised in a series of seven observations which incorporate the relevant known results from the literature (see Fig.~\ref{awesome} for a brief overview of {these}  findings). 

In the second part of this work, we formulate hierarchies of semidefinite relaxations for bounding the set of quantum contextual behaviours and some natural subsets thereof. {Due to the fact that one cannot a priori restrict to} pure quantum states and projective measurements, in general contextuality scenarios standard techniques such as {the} Navascu\'es--Pironio--Ac\'in hierarchy for nonlocal correlations \cite{navascues2008convergent} (or straightforward modification thereof) cannot be employed. Instead, we formulate novel semidefinite programming techniques tailored to the requirements of contextuality scenarios, with particular emphasis on \textit{efficiency} and \textit{ease of implementation}. Of these, our most significant contribution constitutes a semidefinite relaxation technique which employs moment matrices indexed exclusively by monomials of unitary operators, which may be of independent interest. 
We benchmark these relaxations by recovering maximal quantum violations of several noncontextuality inequalities, and identifying optimal quantum protocols with respect to several noncontextuality inequalities in a diverse selection of contextuality scenarios.
Moreover, equipped with these relaxations, we demonstrate the existence of \textit{monogamy of contextuality} in a tripartite setting. Finally, to exhibit the relevance of these relaxations to real-world applications, we present a secure, \emph{semi-device-independent} one-way quantum key distribution scheme powered by preparation contextuality.

%% file: ContextualityScenarios.tex
\section{Contextuality scenarios}

\subsection{Prepare-and-measure experiments}

In this work, we consider \textit{{prepare-and-measure} experiments} consisting of a set of $X$ preparation procedures, $\{P_x\}^{X-1}_{x=0}$, and a set of $Y$ measurement procedures, $\{M_y\}^{Y-1}_{y=0}$. Each measurement is described by $K$ \textit{measurement effects}, $\{ [ k | M_y ] \}_{k=0}^{K-1}$, corresponding to $K$ possible outcomes.
The experiment yields the observed statistics (or \textit{behaviour}), $p(k|x,y) \equiv p(k | P_x, M_y)$, i.e., the probability of observing the outcome $k$, given {that} the measurement $M_y$ was performed on the preparation $P_x$. We may write the statistics as a vector $\mathbf{p} \in \mathbb{R}^{XYK}$.

One may view prepare-and-measure experiments as \textit{one-way communication tasks}, wherein the sender (Alice) encodes her input $x\in[X]$ (where $[X]$ denotes the set $\{0,\ldots,X-1\}$)  {into} the preparations $\{P_x\}_x$, and transmits them to a (spatially separated) receiver (Bob). Bob decodes the message by performing a measurement $M_y$, and obtains an outcome $k \in [K]$. The performance of {a behaviour} in such tasks is gauged by {a} linear functional, called {a} \textit{success metric}. These are of the form $S(\mathbf{p})=\mathbf{c}\cdot \mathbf{p}=\sum_{x,y,k}c_{x,y,k} p(k|x,y)$, where $\mathbf{c}\in\mathbb{R}^{XYK}$ is a vector of real coefficients $c_{x,y,k}$. {Accordingly, the performance of an operational theory in such a one-way communication task is measured by maximising the metric $S$ over the behaviours achievable in the theory.}

\subsection{Operational equivalences and contextuality scenarios} 
First, note that we assume {that} the preparation and measurement procedures in our experiment are modelled by some operational theory in which {probabilistic} mixtures of procedures are possible. Hence the preparations, measurements and effects are {elements} of convex sets, denoted {by} $\mathcal{P}$, $\mathcal{M}$ and $\mathcal{E}$, respectively


Experimental tests of generalised contextuality {require the presence of} equivalences between the preparations  and/or measurement effects. In an operational theory, two preparations, $P_1,P_2\in\mathcal{P}$, are said to be \textit{operationally equivalent}, $P_1 \simeq P_2$, if they yield identical statistics for all measurements and outcomes, i.e., $p(k|P_1,M)=p(k|P_2,M)$ for all effects $[k|M]\in\mathcal{E}$. Similarly, two measurement effects, $[k_1|M_1],[k_2|M_2]\in\mathcal{E}$, are operationally equivalent, $[k_1|M_1] \simeq [k_2|M_2]$, if all preparations assign identical probability of occurrence to them, i.e., $p(k_1|P,M_1)=p(k_2|P,M_2)$ for all $P\in\mathcal{P}$. 

We define a \textit{contextuality scenario} as a prepare-and-measure experiment with certain operational equivalences imposed on the involved preparations and measurements. A scenario is therefore identified by a five-tuple, $T \equiv (X,Y,K,\mathcal{OE}_P,\mathcal{OE}_M)$ where the first three elements specify the number of preparations, measurements, and effects, while $\mathcal{OE}_P$ and $\mathcal{OE}_M$ describe operational equivalences of preparations, and measurement effects, respectively.

A single preparation equivalence in a contextuality scenario consists of two distinct decompositions of an identical \emph{hypothetical}\footnote{{The mixtures of preparations and/or measurement effects that serve to ensure the operational equivalences are referred to as `hypothetical', as they {generally do} not explicitly feature in the experiment.}}
 preparation $Q\in\mathcal{P}$ in terms of the experimental preparations~$P_x$,
\begin{equation}
Q\simeq\sum_{x\in[X]}\alpha_1(x)P_x\simeq\sum_{x\in[X]}\alpha_2(x)P_x,
\end{equation}
where $\alpha_{1}$ and $\alpha_2$ are distinct probability distributions on~$[X]$ corresponding to the two distinct convex decompositions of $Q$. 

In general, a contextuality scenario comprises $V$ distinct hypothetical preparations $\{Q_v\}_{v\in[V]}$, such that each hypothetical preparation $Q_v$ is decomposed into $V_v\geq 2$ distinct mixtures of the experimental preparations $\{P_x\}$ for each $v\in[V]$. Thus, for any given contextuality scenario, all operational equivalences of preparations are summarised as,
\begin{equation}\label{OEP}
   \forall v\in[V], \forall j\in[V_v]: Q_v \simeq\sum_{x\in[X]}  \alpha^v_j(x)P_x,
\end{equation}
where the index $j$ iterates over the $V_v\geq2$ distinct convex decompositions of a given hypothetical preparation $Q_v$, and $\alpha^v_j$ is a probability distribution on $[X]$ specifying the convex decomposition for each $v\in[V]$ and $j\in[V_v]$. Then $\mathcal{OE}_P$ is defined to be the set  $\big\{\{\alpha^v_j\}_{j\in[V_v]}\big\}_{v\in[V]}$, which uniquely specifies the equivalences \eqref{OEP}. 

Likewise, the set $\mathcal{OE}_M =\big\{\{\beta^w_j\}_{j\in[W_w]}\big\}_{w\in[W]}$ determines the operational equivalences of measurement effects\footnote{{Note that the effects of multiple measurements always satisfy a trivial operational equivalence deriving from the completeness of the measurements, i.e., that the probabilities of all the outcomes of a measurement sum to unity. We consider this equivalence part of the definition of a measurement and thus do not include it in the set $\mathcal{OE}_M$. The notation $\mathcal{OE}_M=\emptyset$ therefore signifies the presence of only these trivial measurement equivalences.}}, which are summarised as
\begin{equation} \label{OEM}
\begin{gathered}
    \forall w\in[W],\forall j\in[W_w]: \\  {[k_w|N_w]} \simeq \sum_{k\in[K],y\in[Y]} \beta^w_j(k,y)[k|M_y],
\end{gathered}
\end{equation}
where $\{ {[k_w|N_w]}\}_{w\in[W]}$ is a set of $W$ hypothetical effects in $\mathcal{E}$, the index $j$ iterates over {the} $W_w\geq2$ distinct convex decompositions of a given hypothetical effect $[k_w|N_w]$, and $\beta^j_w$ is a probability distribution on $[K]\times[Y]$, specifying the convex decomposition for each $w\in W$ and $j\in W_w$\footnote{We have formulated the equivalences as in Eqs.~\eqref{OEP} and \eqref{OEM}---as opposed to pairwise---to enable easy distinction between scenarios which can be mapped to Bell scenarios, and those which cannot. Specifically, only contextuality scenarios with a single hypothetical preparation, $V=1$, and only trivial equivalences of measurements, $W=0$ can be mapped to Bell scenarios \cite{schmid2018contextual}.}.

\subsection{Contextual polytope}

The operational equivalences \eqref{OEP} and \eqref{OEM} yield the following constraints on any behaviour~$\mathbf{p}$ appearing in a given contextuality scenario:
\begin{equation}
\begin{gathered}
\label{OEPO}
\forall y\in[Y],k\in[K],v\in[V],\forall j\in[V_v]: \\
q_v=\sum_x \alpha^v_j(x)p(k|x,y)\,, \\ 
\forall x\in [X],w\in[W],\forall j\in[W_w]: \\ 
e_w=\sum_{k,y} \beta^w_j(k,y)p(k|x,y)
\end{gathered}
\end{equation}
where $q_v,e_w \in [0,1]$. As these constraints are linear on the observed behaviour, the set of behaviours $\mathbf{p}$ that satisfy these constraints forms a convex polytope $\mathcal{C} \subset \mathbb{R}^{XYK}$, referred to as the \textit{contextual polytope} (note that this set depends on the tuple $T$ specifying the contextuality scenario, but we will not reflect this in the notation for simplicity). Moreover, for any given contextuality scenario, evaluating the maximal value, $S_{\mathcal{C}}=\max_{\mathbf{p}\in \mathcal{C}} \{ S(\mathbf{p}) \}$, of any success metric $S$ on this set constitutes a linear program \cite{boyd_vandenberghe_2004}.

\subsection{Noncontextual polytope and inequalities} 

A noncontextual ontological model for an operational theory consists of the following {three} elements: (i) a measurable space, $\Lambda$, known as the \emph{ontic state space}, (ii) a probability measure $\mu_P$ on $\Lambda$ describing the \emph{epistemic state} of the system for each preparation $P{\in\mathcal{P}}$, satisfying $\mu_P=\mu_{P'}$ if $P\simeq P'$, and (iii) {for every ontic state $\lambda\in\Lambda$ and measurement $M\in\mathcal{M}$, a probability distribution $\xi_M(\cdot|\lambda)$ over the possible outcomes of $M$, referred to as a \emph{response scheme},} satisfying $\xi_M(k|\cdot)=\xi_{M'}(k'|\cdot)$ if $[k|M]\simeq[k'|M']$.

Consequently, in a given contextuality scenario with operational equivalences of the form \eqref{OEP}, noncontextuality imposes the following constraints on the epistemic states and response schemes,
\begin{equation} \label{OEPOEMMOntic}
\begin{gathered}
\forall v\in[V], \forall j\in[V_v]: \\
\mu_{Q_v}{=}\sum_{x\in[X]}  \alpha^v_j(x)\mu_{P_x}\,, \\ 
\forall w\in[W],\forall j \in[W_w]: \\
{\xi_{N_w}(k_w|\cdot)=}\sum_{k,y} \beta^w_j(k,y)\xi_{M_y}(k|{\cdot}){\,.}
  \end{gathered}
\end{equation}

For any given contextuality scenario, the set of behaviours $\mathbf{p}$ that admit a noncontextual explanation, i.e., for which $p(k|x,y)=\int_\Lambda\xi_{M_y}(k|\lambda)\mu_{P_x}(\lambda)d\lambda$ for some epistemic states and response schemes satisfying the constraints \eqref{OEPOEMMOntic}, forms a convex polytope, $\mathcal{NC} \subset \mathbb{R}^{XYK}$, referred to as the \textit{noncontextual polytope}  \cite{PhysRevA.97.062103}. Consequently, {finding the maximal value, $S_{\mathcal{NC}}=\max_{\mathbf{p}\in \mathcal{NC}} \{ S(\mathbf{p}) \}$, of any success metric $S$ on the noncontextual ontological behaviours}, also constitutes a linear program \cite{PhysRevA.97.062103}. In general, $\mathcal{NC}\subsetneq \mathcal{C}$, which can be witnessed by some success metric. Inequalities of the form $S(\mathbf{p})\leq S_{\mathcal{NC}}$ are typically referred to as \textit{noncontextuality inequalities}, and their violation $S(\mathbf{p}) > S_{\mathcal{NC}}$ implies that the observed behaviour is {contextual}, i.e., $\mathbf{p}\in \mathcal{C} \setminus \mathcal{NC}$. 

\subsection{Quantum contextual set and subsets}

In quantum theory, both preparations and measurement effects are elements of the set of bounded positive semidefinite operators $\mathcal{B}_+(\mathcal{H})$ on some {separable} Hilbert space $\mathcal{H}$ that we assume to be finite dimensional. Preparations are described by \textit{density operators}, 
{ i.e., elements of $\mathcal{B}_+(\mathcal{H})$ with unit trace}.
Measurements are defined by {sets $\{M_k\in\mathcal{B}_+(\mathcal{H})\}_k$ of} 
\textit{measurement operators} {such that $\sum_k M_k = \mathbb{I}$, where $\mathbb{I}$ is the identity operator on $\mathcal{H}$. Such a set is often referred to as a \textit{positive-operator-valued measure} (POVM)}.
The operational equivalences \eqref{OEP} and \eqref{OEM} imply the following constraints on quantum states {$\{\rho_x\}_{x=0}^{X-1}$} and measurements {$\{\{M^y_k\}_{k=0}^{K-1}\}_{y=0}^{Y-1}$ in a contextuality scenario}:
\begin{equation}
    \begin{gathered}
\label{OEPQ} 
\forall v\in[V],\forall j\in[V_v]: \\\sigma_{v}=\sum_x \alpha^v_j(x)\rho_x, \\ 
\forall w\in[W],\forall j\in[W_w]: \\ F_w = \sum_{k,y} \beta^w_j(k,y)M^y_k,
    \end{gathered}
\end{equation}
where $\sigma_v$ are hypothetical quantum states and $F_w$ are hypothetical measurement operators. The set of behaviours with a quantum realisation, i.e., those in which $p(k|x,y)=\tr(\rho_xM^y_k)$ for some quantum states {$\{\rho_x\}_{x=0}^{X-1}$} and measurements {$\{\{M^y_k\}_{k=0}^{K-1}\}_{y=0}^{Y-1}$} satisfying the constraints \eqref{OEPQ}, form the \emph{quantum contextual} set $\mathcal{Q}$. 

It is known that there exist contextuality scenarios in which $\cQ\subsetneq\cC$, specifically those which can be mapped to Bell scenarios \cite{schmid2018contextual} (where this statement is equivalent to quantum theory not being maximally non-local). We demonstrate that this strict inclusion continues to hold in more general scenarios, for example, see Table~\ref{SchidtSpekkens}.

    Apart from $\mathcal{Q}$ itself, we study two natural subsets of it: (i) the set of behaviours realised by pure quantum states $\rho^2_x = \rho_x$, denoted by $\mathcal{Q}^{\Psi}$, and (ii) the set of behaviours realised by projective (sharp) measurements, i.e., those in which $(M^y_k)^2 = M^y_k$ denoted by $\mathcal{Q}^{\Pi}$.

{We now collect some observations regarding the set of quantum behaviours in contextuality scenarios expressed in the notation we have laid out. Some of the{se} observations are already known or partially known, as indicated subsequently.} 

\begin{prop} \label{QisConvex}
For all contextuality scenarios the set of quantum behaviours $\mathcal{Q}$ is convex.
\end{prop}
\begin{prop} \label{QPiisConvex}
For all contextuality scenarios the subset of quantum behaviours $\mathcal{Q}^{\Pi}$ is convex.
\end{prop}
\begin{prop} \label{NCequalsCBezPrepEq}
For all contextuality scenarios with only trivial equivalences of preparations, $T = (X,Y,K,\emptyset,\mathcal{OE}_M)$, all behaviours are noncontextual, i.e., $\mathcal{NC}=\mathcal{Q}=\mathcal{C}$. 
\end{prop}
\begin{prop} \label{NCequalsCLessThanThreePreps}
For all contextuality scenarios with three or fewer preparation procedures, $T = (X \le 3,Y,K,\mathcal{OE}_P,\mathcal{OE}_M)$, all behaviours are noncontextual, i.e., $\mathcal{NC}=\mathcal{Q}=\mathcal{C}$. 
\end{prop}
\begin{prop} \label{QEqualsQPi}
For all contextuality scenarios with only trivial equivalences of measurement effects $T\equiv (X,Y,K,\mathcal{OE}_P,\emptyset )$, projective measurements are sufficient to recover the quantum set, i.e., $\mathcal{Q}=\mathcal{Q}^\Pi$.
\end{prop}
\begin{prop} \label{QNotEqualsQPsi}
In general, for contextuality scenarios with non-trivial equivalences of preparations, pure states cannot produce all quantum contextual behaviours, i.e., $\mathcal{Q}^\Psi \subsetneq \mathcal{Q}$.
\end{prop}
\begin{prop} \label{QNotEqualsQPi}
In general, for contextuality scenarios with non-trivial equivalences of measurement effects, projective measurements cannot produce all quantum contextual behaviours, i.e., $\mathcal{Q}^\Pi \subsetneq \mathcal{Q}$.
\end{prop}

The proofs have been deferred to Appendix~\ref{app:proofs} for brevity. {In Observation~\ref{NCequalsCBezPrepEq}, the fact that $\mathcal{NC}=\mathcal{Q}$ follows from Spekkens' original work \cite{spekkens2005contextuality}. The second part of the result, $\mathcal{Q}=\mathcal{C}$, can be seen as a generalisation (from one preparation to multiple preparations) of the fact that in a hypergraph contextuality scenario all probabilistic models trivially have a quantum realisation when one allows for POVMs, as discussed in Ref.~\cite{kunjwal2019beyond}. In Observation~\ref{NCequalsCLessThanThreePreps}, the fact that $\mathcal{NC}=\mathcal{Q}$ was shown for two binary measurements in Ref.~\cite{SimplePusey} and generally in Ref.~\cite{Noob70747}. Finally, Observation~\ref{QNotEqualsQPi} is implied by the connection between compatibility scenarios and Kochen--Specker contextuality scenarios \cite{CompatibilityToKS} (which can further be translated to generalised contextuality scenarios \cite{KStoSpekkens}), together with the fact that certain compatibility scenarios necessitate unsharp measurements \cite{JointMeasurabilityPoVm}}. These observations yield a rich landscape of contextual behaviours (see Fig.~\ref{awesome}). Due to the fact that the Hilbert space dimension of the quantum systems remains unbounded in contextuality scenarios\footnote{In general, one cannot expect to be able to bound the dimension of the quantum systems necessary for the maximal violation of a given noncontextuality inequality. This is due to known connections between contextuality scenarios and nonlocal scenarios \cite{LIANG20111,schmid2018contextual,PhysRevA.97.062103}, and the fact that there exist nonlocal correlations that require infinite dimensional quantum systems \cite{CS18}.}, evaluating the maximum value of a generic success metric over all quantum realisations, $S_\mathcal{Q}=\max_{\mathbf{p}\in \mathcal{Q}}\{S( \mathbf{p} )\}$, is a remarkably hard problem. However, internal maximisation techniques, such as the see-saw semidefinite programming method, yield efficient dimension-dependent lower bounds, $S_\mathcal{Q} \ge S_{\mathcal{Q}_L}$, for any success metric $S$.

%% file: sdpRelax.tex
\section{Semidefinite programming relaxations of quantum contextual sets}
In order to outer approximate the set $\mathcal{Q}$ and provide upper bounds on success metrics, we formulate \textit{semidefinite relaxation} techniques. A crucial concept underpinning these relaxation techniques is that of \textit{moment matrices}.
For a given density operator, $\rho$, and a sequence of linear operators, $\mathcal{O}=(O_1,O_2,\ldots)$, on $\mathcal{H}$, the \emph{moment matrix hinged on $\rho$}, $\Gamma^{\cO}_\rho$, is a matrix with elements\footnote{{Note that in the index of $(\Gamma^{\cO}_\rho)_{O_j,O_k}$, $O_j$ and $O_k$ represent unique labels of the operators from the sequence $\cO$, not the operators themselves.}} $(\Gamma^{\cO}_\rho)_{O_j,O_k}{\equiv(\Gamma^{\cO}_\rho)_{j,k}}=\tr(\rho {O_j}^\dagger {O_k})$ for $O_j,O_k\in\mathcal{O}$. Due to the positivity of $\rho$, any such moment matrix is positive semidefinite\footnote{Any positive operator $\rho$ can be written as its spectral decomposition $\rho = \sum_k \lambda_k \ketbra{\psi_k}{\psi_k}$, where $\ketbra{\psi_k}{\psi_k}$ is the rank-1 projection onto the span of $\ket{\psi_k} \in \mathcal{H}$, and $\lambda_k \ge 0$ for all $k$. The moment matrix hinged on $\ketbra{\psi_k}{\psi_k}$, $\Gamma^{\cO}_k$, is the Gram matrix of the vectors $\{ O \ket{ \psi_k}\}_{O \in \mathcal{O}}$, and is therefore positive semidefinite. It follows that $\Gamma^{\cO}_\rho = \sum_k \lambda_k \Gamma^{\cO}_k$ is also positive semidefinite.}, $\Gamma^{\cO}_\rho\geq 0$. 

We may use the existence of such moment matrices, satisfying additional constraints, as a necessary condition for a behaviour to be contained in the quantum set in a given contextuality scenario, as follows. Given a contextuality scenario, $(X,Y,K,\mathcal{OE}_P,\mathcal{OE}_M)$, and a behaviour, $\mathbf{p}\in\cQ$, let $\{\rho_x\}_{x\in[X]}$ be density operators and $\{\{M^y_k\}_{k\in[K]}\}_{y\in[Y]}$ be POVMs, on some finite dimensional Hilbert space $\cH$, satisfying \eqref{OEPQ}, respectively, such that $p(k|x,y)=\tr(\rho_xM^y_k)$. It follows that for any finite sequence of linear operators $\cO$ on $\cH$, the moment matrices $\Gamma_x^\cO \equiv \Gamma_{\rho_x}^\cO$ satisfy
\begin{eqnarray} \label{OEPG}
\forall v\in[V], \forall j \in [V_v]: \Theta_v=\sum_x \alpha_j^v(x) \Gamma^{\cO}_x, 
\end{eqnarray}
where $\{\Theta_v\}$ are hypothetical moment matrices hinged on the hypothetical quantum states $\{\sigma_v\}$ from Eq.~\eqref{OEPQ}, which {generally do} not appear in the program\footnote{Notice that the constraints \eqref{OEPG} on the moment matrices form a direct semantic extension of the operational equivalences of preparations in \eqref{OEP} and \eqref{OEPQ}. These constraints are therefore independent of the list of operators indexing the moment matrices.}. 

If we take $\cO$ to be a sequence of monomials of the operators $M^y_k$ (that includes all the length-zero ($\mathbb{I}$) and length-one ($M^y_k$) monomials), then we also find that the moment matrices $\Gamma_x$ satisfy the constraints $(\Gamma^\cO_x)_{\mathbb{I},\mathbb{I}}=1$ and $(\Gamma^\cO_x)_{\mathbb{I},M^y_k}=p(k|x,y)$ for all $x\in[X]$, $y\in[Y]$ and $k\in[K]$. Thus, the existence of positive semidefinite matrices $\Gamma_x^\cO$ for all $x\in[X]$ satisfying these constraints and \eqref{OEPG} is a necessary condition for $\mathbf{p}$ to be in the quantum set, and constitutes a \textit{semidefinite feasibility problem}. {These problems can be solved using \textit{semidefinite programs} (SDP), for which efficient algorithms exist~\cite{boyd_vandenberghe_2004}.} However, using only the constraints above, the diagonal terms of the matrices $\Gamma_x^\cO$ remain unbounded. This unboundedness renders the above necessary condition for $\mathbf{p}$ to be in $\cQ$ trivial. When bounding the set of quantum correlations in Bell scenarios, this issue can be resolved by restricting to projective measurement operators \cite{navascues2008convergent}, resulting in extra constraints on the diagonal elements. 

{However, in contextuality scenarios we cannot assume that our measurement operators are projections, since there are quantum behaviours that cannot be realised with projective measurements (see Observation \ref{QNotEqualsQPi}). This fact poses one of the main challenges in formulating semidefinite relaxations for contextuality scenarios.}
To tackle this issue, we make use of the following lemma\footnote{Note that there exist methods to tackle such scenarios in the literature on non-commuting polynomial optimisation, entailing introduction of additional moment matrices hinged on the positive semidefinite measurement operators $M^y_k$ (localisation) \cite{boyd_vandenberghe_2004,wittek2015algorithm,mironowicz2018applications}, however, the resultant SDPs are rather tedious to implement, and are relatively inefficient.}.
\begin{lem}\label{unitarylemma} Any operator $0 \le M \le \mathbb{I}$ on a finite dimensional Hilbert space $\mathcal{H}$ can be written as $M=\frac{\mathbb{I}}{2}+\frac{U+U^\dagger}{4}$, where $U$ is a unitary operator on $\mathcal{H}$.
\end{lem}
The proof is straightforward, and has been deferred to Appendix \ref{app:proofs}. {This lemma allows us to introduce the following modification of standard semidefinite relaxation techniques.} Instead of taking $\cO$ to be monomials of measurement operators $M^y_k$, we take $\cO$ to be monomials of unitary operators $U^y_k$ (and ${U^y_k}^\dagger$) (along with $\mathbb{I}$) such that $M^y_k=\frac{\mathbb{I}}{2}+\frac{1}{4}(U^y_k+{U^y_k}^\dagger)$. The resulting moment matrices, $\Gamma_x^\cO$, still satisfy \eqref{OEPG}, and additionally we now have 
\begin{equation}\label{eq:gamu}
(\Gamma^\cO_x)_{\mathbb{I},U^y_k}+(\Gamma^\cO_x)_{\mathbb{I},{U^y_k}^\dagger}=4\left(p(k|x,y)-\frac{1}{2}\right)\,
\end{equation}
and
\begin{equation}\label{eq:diag1}
(\Gamma^\cO_x)_{O,O} = 1\,,
\end{equation}
for all $x \in [X]$ and $O \in \mathcal{O}$. Imposing this final condition alleviates the problem of having unbounded diagonal terms.

Firstly, we consider the sequence of monomials of length at most one, i.e., the sequence $\cU_1=(\mathbb{I})\mathbin\Vert (U^y_k,{U^y_k}^\dagger)^{Y-1,K-2}_{y=0,k=0}$ where $\mathbin\Vert$ denotes the concatenation of sequences (due to the relations $\sum_k M^y_k = \mathbb{I}$, we do not need to include the last unitary operators $U^y_{K-1},({U^y_{K-1}})^\dagger$ in $\mathcal{U}_{1}$). In addition to Eqs.~\eqref{OEPG}, \eqref{eq:gamu}, and \eqref{eq:diag1} with $\cO=\cU_1$, the measurement equivalences $\mathcal{OE}_M$ imply the following additional constraints:

\begin{align} \label{OEMG} \nonumber
&\forall x\in [X],O\in \mathcal{U}_{1},w\in[W],\forall j\in[W_w]: \\ \nonumber
&q_{w}=\sum_{k,y} \beta^{w}_j(k,y)((\Gamma^{\cU_1}_x)_{O,U^y_k}+(\Gamma^{\cU_1}_x)_{O,{U^y_k}^\dagger}), \\
&n_w=\sum_{k,y} \beta^{w}_j(k,y)((\Gamma^{\cU_1}_x)_{U^y_k,O}+(\Gamma^{\cU_1}_x)_{{U^y_k}^\dagger,O})\,,
\end{align} 
where $q_v,n_w \in \mathbb{C}$, and we define
\begin{equation}
\begin{split}
& \left. (\Gamma^{\cU_1}_x)_{O,U^y_{K-1}}+ (\Gamma^{\cU_1}_x)_{O,{U^y_{K-1}}^\dagger} \right. \\
= & \left. 2(2-K)(\Gamma^{\cU_1}_x)_{O,\mathbb{I}} \right. \\
& \left.
-\sum^{K-2}_{k=0}\left((\Gamma^{\cU_1}_x)_{O,U^y_{k}}+(\Gamma^{\cU_1}_x)_{O,{U^y_{k}}^\dagger}\right) \right.
\end{split}
\end{equation}
and
\begin{equation}
\begin{split}
& \left. (\Gamma^{\cU_1}_x)_{U^y_{K-1},O}+ (\Gamma^{\cU_1}_x)_{{U^y_{K-1}}^\dagger,O} \right . \\
= & \left. 2(2-K)(\Gamma^{\cU_1}_x)_{\mathbb{I},O}\right.
\\ &\left.
 -\sum^{K-2}_{k=0}\left((\Gamma^{\cU_1}_x)_{U^y_{k},O}+(\Gamma^{\cU_1}_x)_{{U^y_{k}}^\dagger,O}\right) \right.
\end{split}
\end{equation}
for all $O\in\mathcal{U}_{1}$ and $x\in[X]$.

We denote by $\cQ_1$ the set of behaviours $\mathbf{p}$ such that there exist positive semidefinite matrices $\Gamma_x^{\cU_1}$ satisfying Eqs.~\eqref{OEPG}, \eqref{eq:gamu}, \eqref{eq:diag1} and \eqref{OEMG}.
{This problem is again a semidefinite feasibility problem. It is further important to note that maximising a success metric, $S$, over the set $\cQ_1\supset\cQ$ also constitutes an SDP. We will denote this maximum by $S_{\cQ_1}$, and evidently, it is an upper bound on the maximum quantum value, $S_\cQ$.}

{Note that for the feasibility SDP the constraints deriving from normalisation of the probabilities are implicitly included in Eq.~\eqref{eq:gamu}, assuming the behaviour is properly normalised. On the other hand, in the optimisation SDP we can impose such constraints by excluding the operators corresponding to the final outcome of each measurement---as in the feasibility SDP---and imposing the non-negativity of the probabilities, i.e., 
\begin{equation}\frac{1}{2}+\frac14\left(\left(\Gamma^{\mathcal{U}_1}_x\right)_{\mathbb{I},U^y_k}+\left(\Gamma^{\mathcal{U}_1}_x\right)_{\mathbb{I},{U^y_k}^\dagger}\right)\geq 0,
\end{equation}
for all $k\in[K-2]$ and $x\in[X]$, along with their subnormalisation,
\begin{equation}\sum_{k=0}^{K-2}\left[\frac{1}{2}+\frac14\left(\left(\Gamma^{\mathcal{U}_1}_x\right)_{\mathbb{I},U^y_k}+\left(\Gamma^{\mathcal{U}_1}_x\right)_{\mathbb{I},{U^y_k}^\dagger}\right)\right]\leq 1,
\end{equation}
for all $x\in[X]$. The success metric should also be re-expressed using $p(K-1|x,y)=1-\sum^{K-2}_{k=0}p(k|x,y)$.}

The upper bounds $S_{\cQ_1}\geq S_{\cQ}$ are in general non-trivial, as we will demonstrate in Sec.~\ref{sec:tight}. Moreover, in many of the cases we consider, these bounds coincide with the lower bounds obtained from see-saw methods, i.e., $S_{\cQ_1}=S_{\cQ_L}$, and are therefore tight (up to machine precision). To demonstrate the ease of implementation and high efficiency of this \textit{unitary-based} SDP relaxation, we provide a code tutorial in Appendix \ref{app:tutorial}.

The set $\cQ_1$ can be considered as the first \emph{level} of a hierarchy of sets characterised by an \emph{SDP hierarchy}. Level one of an SDP hierarchy employs moment matrices with some operator sequence $\mathcal{O}_{1}$, e.g., $\cO_1=(\mathbb{I},A,B)$. Then level $l\in\mathbb{N}$ employs moment matrices with operator sequence $\cO_l$ consisting of monomials of elements of $\cO_1$ with length at most $l$ , e.g., $\mathcal{O}_{2} =(\mathbb{I},A,B,AB,BA)$. Accordingly, we may define a hierarchy of outer approximations of the quantum set, $\cQ_1\supseteq \cQ_2\supseteq \ldots \supseteq \cQ$, characterised by the unitary-based SDP hierarchy in which level $l$ is characterised by the operator sequence $\cU_l$.

Next, in order to bound the set $Q^{\Pi}$, we revert to the standard projection-based SDP relaxations involving operator lists composed of projections. Specifically, we consider moment matrices $\left\{\Gamma^{\mathcal{P}_{1}}_x\right\}_x$ with operator sequence $\mathcal{P}_{1} = (\mathbb{I})\mathbin\Vert(\Pi^k_y)^{Y-1,K-2}_{y=0,k=0}$, where $\Pi^y_k$ are projections. Apart from \eqref{OEPG}, and the trivial constraints $(\Gamma^{\mathcal{P}_{1}}_x)_{\mathbb{I},\mathbb{I}}=1$ for all $x\in [X]$, we have additional constraints that follow from the projectivity of measurement operators, $(\Gamma^{\mathcal{P}_{1}}_x)_{\Pi^y_k,\Pi^y_{k'}}= \delta_{k,k'} (\Gamma^{\mathcal{P}_{1}}_x)_{\mathbb{I},\Pi^y_k}$ for all $x\in [X]$, $y\in[Y]$ and $k,k' \in [K]$. Further constraints emerge from operational equivalences of measurement effects (similar to \eqref{OEMG}). Hence, for all contextuality scenarios, one can define the hierarchy of sets $\mathcal{Q}_{l+1}^\Pi \subseteq \mathcal{Q}_l^\Pi$ for all $l \in \mathbb{N}$, containing $\mathcal{Q}^\Pi$. 

Finally, in order to bound the set $Q^{\Psi}$, we need not even invoke multiple moment matrices, instead we employ an \textit{unhinged} moment matrix $(\Gamma^{\mathcal{PS}_1})_{O,O'}=\tr({O}^\dagger{O'})$ with operator sequence $\mathcal{PS}_1=(\mathbb{I}) \mathbin\Vert(\ketbra{\psi_x})^{X-1}_{x=0}\mathbin\Vert(\Pi^k_y)^{Y-1,K-2}_{y=0,k=0}$.  It is straightforward to see that such a moment matrix is necessarily positive semidefinite. In addition to the constraints similar to the ones described above, we have constraints of the form $(\Gamma)_{\mathbb{I},\ketbra{\psi_x}}=(\Gamma)_{\ketbra{\psi_x},\ketbra{\psi_x}}=1$ that follow from the purity of the quantum preparations, along with constraints that ensue from operational equivalences of preparations (similar to \eqref{OEMG}). We refer to this formulation as the \textit{pure-state based} SDP relaxation. For all contextuality scenarios, we may define the hierarchy of sets $\mathcal{Q}_{l+1}^\Psi \subseteq \mathcal{Q}_l^\Psi$ for all $l \in \mathbb{N}$, containing $\mathcal{Q}^\Psi$.

%% file: TightBounds.tex
\section{Tight bounds on the quantum set}\label{sec:tight}
The primary application of the relaxations introduced in the previous section is finding numerical proofs for the maximal quantum value $S_Q$ of any success metric $S$. Such a proof is possible when the lower bound $S_{\mathcal{Q}_L}$ obtained from internal maximisation techniques, such as {the} see-saw semidefinite programming method \cite{Noob1,mironowicz2018applications}, matches the upper bound obtained from {some finite level $l$} of the suitable semidefinite relaxation hierarchy, i.e., $S_{\cQ_L}=S_{\cQ_l}$ (up to machine precision).

We begin by considering the simplest non-trivial single parameter family of contextuality scenarios \cite{SimplePusey}, $(4,2,2,\mathcal{OE}_P(\alpha),\emptyset)$, wherein the preparation setting $x$ is composed of two bits $x= (x_0,x_1) \in \{0,1\}^2$, and we consider a family of operational equivalence conditions of the form $\frac{1}{2}P_{00}+\frac{1}{2}P_{11}\simeq \alpha P_{01}+(1-\alpha)P_{10}$, parameterised by the coefficient $\alpha\in[0,1]$, along with the success metric {$S^{\textit{rac}}(\mathbf{p})=\frac{1}{8}\sum_{x,y}p(x_y|x,y)$}. In light of Observation \ref{QEqualsQPi}, for preparation contextuality scenarios, we employ the projection-based relaxation $\mathcal{Q}^\Pi_l$, as for such scenarios one can in general take the measurement operators to be projections. While the first level $\mathcal{Q}^\Pi_1$ retrieves non-trivial upper bounds {$S^{\textit{rac}}_{\mathcal{Q}^\Pi_1}<S^{\textit{rac}}_{\mathcal{C}}$}, the bounds obtained from the second level of the relaxation $\mathcal{Q}^\Pi_2$ saturate the lower bounds obtained by internal maximisation, thus we find {$S^{\textit{rac}}_Q = S^{\textit{rac}}_{\mathcal{Q}_L}=S^{\textit{rac}}_{\mathcal{Q}^\Pi_2}$}. 
\begin{figure}
    \centering
    \includegraphics[width=\linewidth]{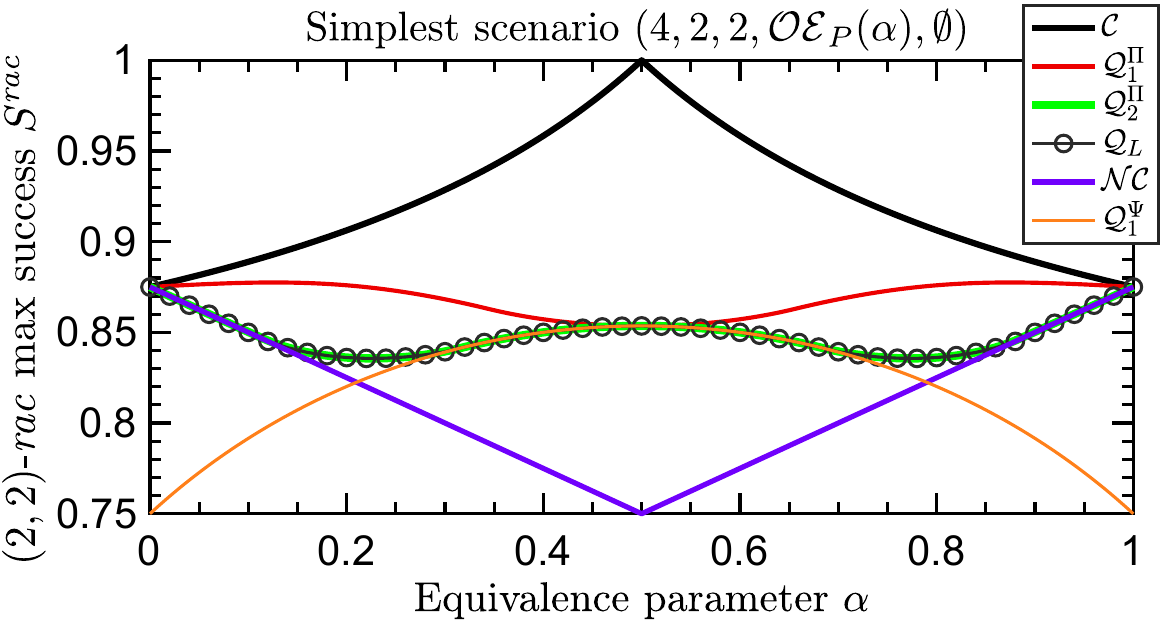}
    \caption{\label{batman} {Plot of bounds on the maximal value of $S^\textit{rac}(\mathbf{p})$ for behaviours $\mathbf{p}$ in various sets in the scenario $(4,2,2,\mathcal{OE}_P(\alpha),\emptyset)$}. Here the solid black curve represents the {tight upper bound, $S^\textit{rac}_\cC$, for the polytope $\cC$ of contextual behaviours.} The solid red {and green} curve{s} {represent the upper bounds, $\mathcal{Q}^\Pi_1$ and $\mathcal{Q}^\Pi_2$, on $S^\textit{rac}_Q$ given by} the first {and second} level{s} of the projection-based relaxation{, respectively.}  {T}he circled black curve represents the lower bounds {$S^\textit{rac}_{Q_L}$ on $S^\textit{rac}_Q$} obtained via the see-saw semidefinite programming method, the solid orange curve corresponds to the {upper bounds $S^\textit{rac}_{Q^\Psi_1}$ on the values $S^\textit{rac}_{Q^\Psi}$ obtained from the }first level of the pure-state based relaxation, and the solid purple curve corresponds to the {tight upper bounds $S^\textit{rac}_{\mathcal{NC}}$ for the} noncontextual polytope $\mathcal{NC}$. As the {lower} bounds {on the value $S^\textit{rac}_Q$} from internal maximization coincide ({up to} numerical precision) with the upper bounds from the second level of the projection-based relaxation, we obtain proofs of optimality, $S^{\textit{rac}}_{\cQ} = S^{\textit{rac}}_{\cQ_L} = S^{\textit{rac}}_{\mathcal{Q}^{\Psi}_{1}}$.
    Notice for a range of the scenario parameter $\alpha$, $S^{\textit{rac}}_\mathcal{Q}>S^{\textit{rac}}_{\mathcal{Q}^\Psi_1}\geq S^{\textit{rac}}_{\mathcal{Q}^\Psi}$, which constitutes a proof of Observation \ref{QNotEqualsQPsi}, demonstrating that certain extremal quantum contextual behaviours {can only be realised through the explicit use of mixed quantum states in the experiment.}}
\end{figure}
Moreover, for this case we employed the pure-state based relaxation $\mathcal{Q}^{\Psi}_1$ to bound the set $\mathcal{Q}^{\Psi}$. Remarkably, for a range of the scenario parameter $\alpha$, we find that pure states are not sufficient to violate the noncontextual bound, i.e., {$S^{\textit{rac}}_{\mathcal{Q}^{\Psi}}\leq S^{\textit{rac}}_{\mathcal{Q}^{\Psi}_1}< S^{\textit{rac}}_\mathcal{NC} < S^{\textit{rac}}_{\mathcal{Q}}$}, which constitutes a numerical proof of Observation \ref{QNotEqualsQPsi} (see Fig.~\ref{batman}). 

{It is helpful to recognise that the metric $S^{\textit{rac}}$ gives the average success probability when the contextuality scenario is cast as the following communication task. A party (Alice) encodes two random bits, $x=(x_0,x_1)$, in a preparation $P_x$ and sends it to a second party (Bob). Bob must guess the bit $x_y$, where $y\in\{0,1\}$ is chosen at random. He does so by performing a measurement $M_y$ on the preparation, the outcome of which constitutes his guess. The probability of winning given $x$ and $y$ is therefore $p(x_y|P_x,M_y)$. This task is an example of \emph{random access coding}. In general, we refer to an analogous task in which Alice must encode $n$ dits into her preparations, subject to some constraints, as an $(n,d)$-\textit{rac}. The relevant contextuality scenario for an $(n,d)$-\textit{rac} is $(d^n,n,d,\mathcal{OE}_P,\mathcal{OE}_M)$, for some operational equivalences $\mathcal{OE}_P$ and $\mathcal{OE}_M$.} 

Next, we consider a cryptographically appealing family of preparation noncontextuality inequalities based on {$(n,2)$-\textit{racs} that impose \emph{parity obviliousness}, known as} {$(n,2)$-\textit{poracs}, and} introduced in Ref.~\cite{PhysRevLett.102.010401}. {Parity obliviousness is the cryptographic requirement that Bob must gain no information about the \emph{parity} of any subset (of order greater than 1) of Alice's bits. Explicitly, if Alice's bits are $ (x_0,\ldots,x_{n-1})$, she must not send any information about $\oplus_{i\in J}x_{i}$ for any $J\subseteq[N]$ such that $|J|>1$, where $\oplus$ denotes addition modulo two.} {Parity obliviousness} necessitates nontrivial operational equivalences{, $\mathcal{PO}(n,2)$, on the preparations}, thereby tying these communication tasks to contextuality scenarios $(2^n,n,2,\mathcal{PO}(n,2),\emptyset)$. For instance, consider the $(3,2)$-\textit{porac} task {and denote Alice's three input bits by} $x = (x_0,x_1,x_2)$ {and the corresponding preparations by $P_x$. Here parity obliviousness} translates to the following operational equivalences, $\mathcal{PO}(3,2)$:
\begin{equation} \label{32PORACOEPPair}
    \frac12( P_{00x_2} + P_{11x_2}) \simeq \frac12( P_{01x_2} + P_{10x_2} ),
\end{equation}
{for all $x_2\in\{0,1\}$}, and similarly for the other pairs of bits, {and}
\begin{equation} \label{32PORACOEPThree}
\begin{split}
    & \left. \frac14 ( P_{000} + P_{011} + P_{101} + P_{110} ) \right. \\
    \simeq & \left. \frac14 (P_{001} + P_{010} + P_{100} + P_{111} ). \right.
\end{split}
\end{equation}
{The average success probability of a strategy in the $(n,2)$-\textit{porac} is given by the success metric $S^{\textit{rac}}(\mathbf{p})=\frac{1}{2^n n}\sum_{x,y}p(x_{y}|x,y)$ on its behaviour $\mathbf{p}$ in the contextuality scenario $(2^n,n,2,\mathcal{PO}(n,2),\emptyset)$}. The noncontextual {upper} bound on the success metric is $S^{\textit{rac}}_{\mathcal{NC}}=\frac{1}{2}(1+\frac{1}{n})$. The first levels of both relaxations, the {projection}-based $\mathcal{Q}^{\Pi}_1$ and the unitary-based $\mathcal{Q}_1$, {match} the known optimal quantum bounds for $n\in\{2,\ldots,7\}$, i.e., $S^{\textit{rac}}_\mathcal{Q}=S^{\textit{rac}}_{\mathcal{Q}^{\Pi}_1}=S^{\textit{rac}}_{\mathcal{Q}_1}=\frac{1}{2}(1+\frac{1}{\sqrt{n}})$ {up to numerical precision} \cite{chailloux2016optimal,ghorai2018optimal}. 

{Thus far} we have considered contextuality scenarios {with only trivial} equivalences of measurement effects. Moving to tests of universal contextuality, we consider the scenario $(6,3,2,\mathcal{U}_P,\mathcal{U}_M)$ where we have the operational equivalences $\frac{1}{2}(P_{0}+P_{1}) \simeq \frac{1}{2}(P_{2}+P_{3})\simeq \frac{1}{2}(P_{4}+P_{5})$, and $\frac{1}{3}([0|M_0]+[0|M_1]+[0|M_2]) \simeq \frac{1}{3}([1|M_0]+[1|M_1]+[1|M_2])$. Here the noncontextual polytope has six distinct non-trivial facets up to relabelling symmetries \cite{PhysRevA.97.062103}, each corresponding to a distinct noncontextuality inequality. We give the maximal quantum violation of all these facet inequalities in Table \ref{SchidtSpekkens}, obtained from the first level of the unitary-based relaxation. Moreover, we retrieve optimal states and measurements that saturate the bounds, displayed in Fig.~\ref{StatesAndMeas}. 
\input{Tables/SchmidtSpekkensIneqTable}

Next, we considered universal noncontextuality inequalities obtained from $n$-\textit{cycle} Kochen--Specker {(KS)} graphs \cite{PhysRevA.94.062103,LIANG20111}. In this case also, {level one of} the unitary-based relaxation retrieves tight quantum bounds (up to machine precision), $S^{ks}_{\mathcal{Q}}=3+\frac{n\cos{\frac{\pi}{n}}}{1+\cos{\frac{\pi}{n}}}$ for $n\in\{5,7,9,11\}${, where $S^{ks}$ is the success metric given in Ref.~\cite[Section~III.C]{PhysRevA.94.062103}} (see Table \ref{KCnCycle}).
\input{Tables/KSnCycle}

 Finally, we considered a class of novel universal noncontextuality inequalities, termed $(n,2)$-\textit{mporacs}, based on the $(n,2)$-\textit{porac}. In addition to the preparation equivalence conditions that ensue from parity obliviousness, we introduce measurement equivalences {given by} $\frac{1}{n}\sum_{y\in [n]} [0|M_y] \simeq \frac{1}{n}\sum_{y\in [n]} [1|M_y]$.
While the noncontextual bounds on the success metric {$S^\textit{rac}$} remain unaltered {from the $(n,2)$-\textit{porac} case} (see Lemma \ref{mporac}), the quantum bounds are further restricted. Again, the first level of the unitary-based relaxation yields {tight upper} bounds {on the maximal quantum value, $S^\textit{rac}_Q$,} for $n\in\{2,\ldots,7\}$ (see Table \ref{mporacTable}).  
\input{Tables/mporac}

%% file: Tables/SchmidtSpekkensIneqTable.tex
\begin{table}[h]
\resizebox{\columnwidth}{!}{%
\begin{tabular}{@{}cccccc@{}}
\toprule
$S(\mathbf{p})$                                              & $S_{\mathcal{C}}$ & $S_{\mathcal{NC}}$ & $S_{\mathcal{Q}_L}$ & $S_{\mathcal{Q}_1}$ & $S_{\mathcal{Q}_2}$ \\ \midrule
$p_{00} +p_{12} +p_{24} $                        & $3$               & $2.5$              & $3$                 & $3$   & $3$              \\
$p_{00} +p_{11} +p_{24}$                         & $3$               & $2.5$              & $2.8660254$         & $2.8660254$  & $2.8660254$     \\
$p_{00} -p_{02} - 2 p_{04}-2 p_{11} +2 p_{12} +2 p_{24}$ & $4.5$ & $3$ & $3.9209518$ & $3.9209518$ & $3.9209518$ \\
$2p_{00} - p_{11} + 2 p_{12}$                    & $3.5$             & $3$                & $3.3660254$         & $3.3660254$     & $3.3660254$    \\
$p_{00}  - p_{04} +   p_{11} + p_{12} +2 p_{24}$ & $5$               & $4$                & $4.6889010$         & $4.6889010$    & $4.6889010$     \\
$p_{00}  - p_{04} +   2p_{11}  +2 p_{24}$        & $5$               & $4$                & $4.6457513$         & $4.6457513$   & $4.6457513$      \\
$p_{00} -p_{03} - 2 p_{04}-2 p_{11} +2 p_{12} +2 p_{24}$ & $4.5$ & $3.5$ & $3.5$ & $3.5552760$ & $3.5$ \\ \bottomrule
\end{tabular}
}
\caption{\label{SchidtSpekkens} Maximal values for the six facet inequalities \cite{PhysRevA.97.062103} in the universal contextuality scenario $(6,3,2,\mathcal{U}_P,\mathcal{U}_M)$ \cite{schmid2018contextual}, where $p_{yx}=p(0|x,y)$, for contextual behaviours {($S_\mathcal{C}$)}, noncontextual behaviours {($S_\mathcal{NC}$)}, along with lower bounds {on the maximal quantum values $S_{\mathcal{Q}}$} obtained from the see-saw semidefinite programming method {($S_{\cQ_L}$)}, and upper-bounds {on $S_\mathcal{Q}$}  obtained from the first {and second} level{s} of the unitary-based relaxation ($S_{\mathcal{Q}_1}$ {and $S_{\mathcal{Q}_2}$)}. Notice that for the six facet inequalities the lower bounds, $S_{\cQ_L}$, {coincide with} the upper bounds, $S_{\cQ_1}$, up to seven decimal places, which provides a numerical proof for the optimal quantum value, $S_{\cQ}$. Furthermore, while for the first facet inequality the quantum bound saturates the contextual bound, $S_{\mathcal{Q}}=S_{\mathcal{C}}$, for all other inequalities quantum behaviours fail to reproduce extremal contextual behaviours, i.e., $S_{\mathcal{Q}}<S_{\mathcal{C}}$. The seventh inequality is not a facet inequality, but a contextuality inequality stemming from a typing error in Ref.~\cite{PhysRevA.97.062103}, wherein it is stated to be a facet inequality \footnote{The error was confirmed by private communication and the noncontextual bound of $3.5$ for this success metric was calculated only in the present work.}. We include this inequality in our results table, since it presents an example of the quantum upper bound upon restricting ourselves to projective measurements, {$S_{\mathcal{Q}^\Pi_1}=S_{\mathcal{Q}^\Pi_L}=3.4641016$}, being lower than even the noncontextual bound, $S_\mathcal{NC}=3.5$, thus forming a numerical proof of Observation \ref{QNotEqualsQPi}. The seventh inequality also presents a case in which the second level of our hierarchy was required to find a tight bound.}
\end{table}

%% file: Tables/KSnCycle.tex
\begin{table}[h]
\begin{center}
\begin{tabular}{@{}ccccc@{}}
\toprule
$n$  & $S^{ks}_{\mathcal{C}}$ & $S^{ks}_{\mathcal{NC}}$ & $S^{ks}_{\mathcal{Q}_L}$ & $S^{ks}_{\mathcal{Q}_1}$ \\ \midrule
$5$  & $6$               & $5$                & $5.2360679$         & $5.2360679$         \\
$7$  & $7$               & $6$                & $6.3176672$         & $6.3176672$         \\
$9$  & $8$             & $7$                & $7.3600895$         & $7.3600895$         \\
$11$ & $9$             & $8$                & $8.3863029$         & $8.3863029$         \\ \bottomrule
\end{tabular}
\end{center}
\vspace{-0.3cm}
\caption{\label{KCnCycle} Maximal values  for generalised noncontextuality inequalities based on $n$-cycle KS graphs for contextual behaviours ($S^{ks}_\mathcal{C}$), noncontextual behaviours ($S^{ks}_\mathcal{NC}$), along with lower bounds obtained from the see-saw semidefinite programming method ($S^{ks}_{\mathcal{Q}_L}$) and upper bounds obtained from the first level of the unitary-based relaxation ($S^{ks}_{\mathcal{Q}_1}$). Again, in each case the lower bounds ($S^{ks}_{\cQ_L}$) coincide with the upper bounds ($S^{ks}_{\cQ_1}$) up to seven decimal places, which provides a numerical proof for the optimal quantum value, $S^{ks}_{\cQ}$.}
\end{table}

%% file: Tables/mporac.tex
\begin{table}[h]
\begin{center}
\begin{tabular}{@{}ccccc@{}}
\toprule
$n$ & $S^\textit{rac}_{\mathcal{C}}$ & $S^\textit{rac}_{\mathcal{NC}}$ & $S^\textit{rac}_{\mathcal{Q}_L}$ & $S^\textit{rac}_{\mathcal{Q}_1}$ \\ \midrule
$2$ & $3/4$             & $3/4$              & $0.75$              & $0.75$              \\
$3$ & $3/4$             & $2/3$              & $0.75$              & $0.75$              \\
$4$ & $3/4$             & $5/8$              & $0.7165063$         & $0.7165063$         \\
$5$ & $3/4$             & $3/5$              & $0.7041241$         & $0.7041241$         \\
$6$ & $3/4$             & $7/12$             & $0.6863389$         & $0.6863389$         \\
$7$ & $3/4$             & $4/7$              & $0.6767766$         & $0.6767766$         \\ \bottomrule
\end{tabular}
\end{center}
\vspace{-0.3cm}
\caption{\label{mporacTable} Maximal values for the success probability in the $(n,2)$-\textit{mporac} tasks for contextual behaviours ($S^\textit{rac}_\mathcal{C}$), noncontextual behaviours ($S^\textit{rac}_\mathcal{NC}$), along with lower bounds obtained from the see-saw semidefinite programming method ($S^\textit{rac}_{\mathcal{Q}_L}$), and upper bounds obtained from the first level of the unitary-based relaxation ($S^\textit{rac}_{\mathcal{Q}_1}$). Again, in each case the lower bounds ($S_{\cQ_L}$) coincide with the upper bounds ($S^\textit{rac}_{\cQ_1}$) up to seven decimal places, which provides a numerical proof for the optimal quantum value, $S^\textit{rac}_{\cQ}$, and demonstrates the efficacy of the unitary-based relaxation. }
\end{table}

%% file: Applications.tex
\section{Applications} 
\subsection{Monogamy of contextuality}\label{sec:monogamy}

A new feature that we uncover using the relaxations described in the previous sections is that of \textit{monogamy of contextuality}. Consider a prepare-and-measure scenario with two measuring parties, Bob and Charlie, whose actions are space-like separated. The preparation party, Alice, shares a preparation according to her setting $x$ between Bob and Charlie. The observed statistics in such a case are described by $p_{BC}(k_B, k_C | x, y_B, y_C)$, where $y_B$ ($y_C$) and $k_B$ ($k_C$) are the measurement settings and outcomes of Bob (Charlie), respectively. The space-like separation implies \textit{no-signalling} between Bob and Charlie~\cite{nonlocality}, that is,
\begin{equation}
    \begin{aligned}
          p_B(k_B | x, y_B, y_C ) = p_B&(k_B | x, y_B, y'_C ) 
        \\ &\forall k_B, x, y_B, y_C, y'_C \\
         p_C(k_C | x, y_B, y_C ) = p_C&(k_C | x, y'_B, y_C )\\
         &\forall k_C, x, y_B, y'_B, y_C, 
    \end{aligned}
\end{equation}
where $p_B$ and $p_C$ are the marginal distributions of Bob and Charlie, respectively.

We consider the tripartite version of the $(3,2)$-\textit{porac} task, wherein Alice's preparations adhere to the operational equivalences of preparations ensuing from parity obliviousness \eqref{32PORACOEPPair}, \eqref{32PORACOEPThree}, and the two measuring parties, Bob and Charlie, are both aiming to maximise their average success probabilities, $S^\textit{rac}_B$ and $S^\textit{rac}_C$, respectively, expressed by

\begin{equation} 
    S^\textit{rac}_Z = \frac{1}{24} \sum_{x, y_Z} p_Z(x_{y_Z} | x, y_Z ),
\end{equation}
where $Z$ denotes either $B$ or $C$, and $p_Z$ is the respective marginal probability distribution. 
We find that, say, Charlie's maximum success probability is limited by that of Bob, i.e., their maximum success probabilities satisfy certain \textit{monogamy relations}. For any behaviour in $\cC$ also satisfying no-signalling (NS), we obtain the relation $S^\textit{rac}_B + S^\textit{rac}_C \le \frac32$ for the $(3,2)$-\textit{porac} (see Fig.~\ref{fig:monogamy}). {Note that this relation is more restrictive than simply adding up the maximal success probabilities attainable in $\cC$, $S^\textit{rac}_\cC = 1$, hence the term \textit{monogamy}.}

Quantum behaviours in this setup, $\pp \in \cQ$, are described by $p_{BC}(k_B, k_C | x, y_B, y_C) = \tr[ \rho_x (M^{y_B}_{k_B} N^{y_C}_{k_C})]$, where $\{\rho_x\}$ are potentially entangled states on $\cH_B\otimes \cH_C$, $\{M^{y_B}_{k_B}=M'^{y_B}_{k_B}\otimes\mathbb{I}_C\}$ and $\{N^{y_C}_{k_C}=\mathbb{I}_B\otimes N'^{y_C}_{k_C}\}$ are the measurement operators of Bob and Charlie, respectively, and one might assume without loss of generality that they are projections, in light of Observation \ref{QEqualsQPi}. In order to obtain monogamy relations in the quantum case, we employ our projection-based semidefinite relaxation for the set $\cQ^\Pi$, and implement the relaxation $\cQ^\Pi_{1 + BC}$\footnote{The set $\cQ^\Pi_{1 + BC}$ emulates the intriguing physical scenario wherein Alice distributes \emph{almost quantum correlations} \cite{almost_quantum} to Bob and Charlie.}, built from monomials of the projections of Bob, Charlie, and products of Bob's and Charlie's projections, $M^{y_B}_{k_B} N^{y_C}_{k_C}$ (but not employing monomials corresponding to products of different operators of Bob or Charlie only, e.g.~$M^{y_B}_{k_B} M^{y'_B}_{k'_B}$). Besides the constraints that ensue from parity obliviousness, the relaxation $\cQ^\Pi_{1 + BC}$ includes commutation relations (CR)~\cite{navascues2008convergent} that ensue from the spatial separation between Bob and Charlie. 
This relaxation gives us the monogamy relation $S^\textit{rac}_B + S^\textit{rac}_C \le 1.392$. For the complete monogamy relation, see Fig.~\ref{fig:monogamy}, where we plot the maximum attainable $S^\textit{rac}_C$ given a fixed $S^\textit{rac}_B$.

\begin{figure}[h]
    \centering
    \includegraphics[width=\linewidth]{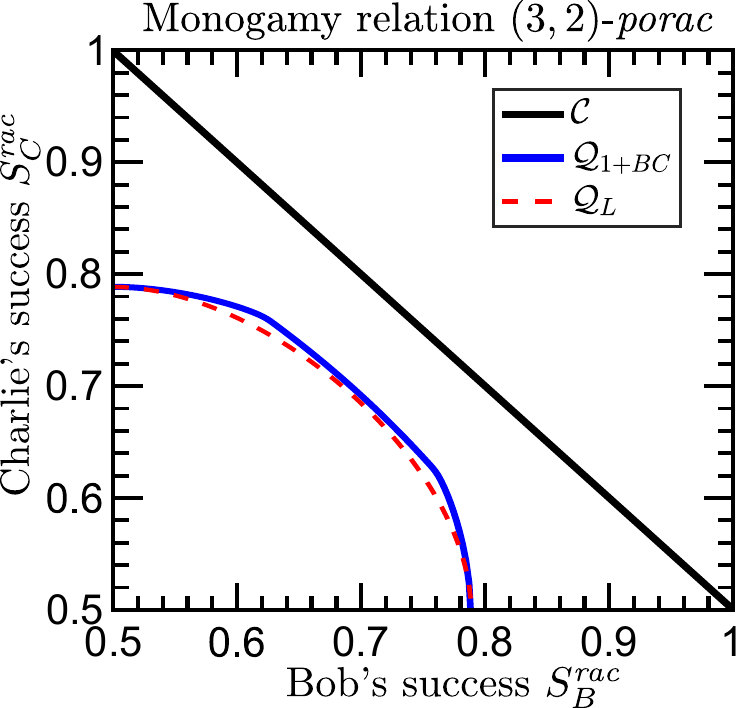}
    \caption{Monogamy relations expressed as bounds on the success probability $S^\textit{rac}_C$ of Charlie, given the success probability $S^\textit{rac}_B$ of Bob. For contextual behaviours $\pp\in\cC$, the upper and lower bounds agree (as computing the value is a simple linear program), and give the solid black curve on the figure. For the quantum set $\cQ$, we obtain upper bounds (solid blue curve) from the ${1+BC}$ level of the projection-based SDP hierarchy with additional commutation relations (CR), and lower bounds (dashed red curve) from the see-saw optimisation method ($\cQ_L$). \label{monogamy}}
    \label{fig:monogamy}
\end{figure}

\subsection{Semi-device-independent quantum key distribution}
{We introduce} the following semi-device-independent quantum key distribution protocol based on the $(3,2)$-\textit{porac}: Alice and Bob play many rounds of the (3,2)-{\textit{porac}} game, in each round uniformly randomly choosing their settings, $x=(x_0,x_1,x_2)$ and $y$, and each of Bob's outcomes forms a bit of his raw key, $b=k$. After accumulating statistics from many rounds, Bob announces his input setting, $y$, for each round via a classical public channel (that can be read but not modified by any third party). After learning these values, Alice only keeps her bit $a = x_y$ to create her raw key, and discards the other two. Bob then announces a part of his output bits, to allow for the parties to estimate the success metric of the {\textit{porac}} task. {Based on this estimate, they either abort the protocol---if the success metric does not reach a certain threshold---or proceed to extract perfectly correlated key bits from their remaining bits $b$ and $a$, by means of one-way classical communication.}

The monogamy relations {derived in the previous section} imply the potential application of preparation contextuality to quantum key distribution. In particular, we might replace the trusted party Charlie with a malicious eavesdropper Eve, who is trying to read the shared key. In a particularly paranoid scenario, we might assume that Eve is responsible for the manufacturing of the devices of Alice and Bob, and therefore she has perfect control and knowledge of the states $\rho_x$ that Alice is able to prepare, as well as the measurements $M^{y}_{b}$ that Bob is able to perform. However, we assume that Alice's preparations, whatever they may be, respect the operational equivalences necessary to guarantee parity obliviousness. Further, for simplicity, we assume that the devices behave in the same way every round (\textit{i.i.d.~assumption} \cite{iid}), and that while running the {\textit{porac}} protocol, Alice and Bob have free will and can choose their settings $x$ and $y$ independently of Eve (\textit{free will assumption} \cite{DIQKD}). We further assume that Eve measures her part of the quantum state in each round---assuming \textit{no quantum memory} of Eve \cite{iid}---and produces a classical output $e$. The resulting general tripartite correlation is of the form $p_{ABE|XY}(a,b,e|x,y) = \delta_{a,x_y} \tr[ \rho_x (M^y_b \otimes E_e) ]$, where Alice's preparations $\{\rho_x\}$ are potentially entangled states on $\cH_B\otimes \cH_E$. Averaging over all the setting choices, the parties end up with the tripartite probability distribution,
\begin{equation}
    p_{ABE}(a,b,e) = \sum_{xy} p_{XY}(x,y) p_{ABE|XY}(a,b,e|x,y),
\end{equation}
and we will assume the uniform distribution $p_{XY}(x,y) = \frac{1}{24}$ for all $x$, $y$.

According to the seminal result by Csisz\'ar and K\"orner~\cite{csiszarkorner}, the extractable key, $r$, between Alice and Bob from the above distribution is lower bounded by
\begin{equation}
    r \ge I(A:B) - I(A:E) = H(A|E) - H(A|B),
\end{equation}
where $I(A:B)$ is the mutual information of $A$ and $B$, $H(A|B)$ is the conditional Shannon entropy, and these quantities are computed from the bipartite marginal distributions.

We bound the first term via the min-entropy $H(A|E) \ge H_\text{min}(A|E)$, where
\begin{widetext}
\begin{equation}\label{eq:Hminbound}
\begin{split}
    H_{\text{min}}(A|E) = & \left. -\sum_e p_E(e) \log\big[ \max_a \{ p_{A|E}(a|e) \} \big] = -\sum_e p_E(e) \log\big[ p_{A|E}(e|e) \big] \right. \\
    \ge & \left. -\log\Big[ \sum_e p_E(e) p_{A|E}(e|e) \Big] = -\log\Big[ \sum_e p_{AE}(e,e) \Big] \right. \\
    = & \left. -\log \Big[ \sum_{e,x,y} p_{XY}(x,y) p_{AE|XY}(e,e|x,y) \Big] = -\log \Big[ \frac{1}{24} \sum_{e,x,y} p_{E|XY}(e|x,y) \cdot \delta_{e,x_y} \Big] \right. \\
    = & \left. -\log\Big[ \frac{1}{24} p_{E|XY}(x_y|x,y) \Big] = -\log S^\textit{rac}_E, \right.
\end{split}
\end{equation}
\end{widetext}
where {$H_{\text{min}}$ is the min-entropy and} $S^\textit{rac}_E$ is {Eve's success probability in the same \textit{porac} task that Alice and Bob perform.} In Eq.~\eqref{eq:Hminbound}, the second equality is assuming an optimal guessing strategy for Eve, and the inequality follows from the concavity of the logarithm.

The term $H(A|B)$ can always be computed from the observed behaviour, which Alice and Bob are able to estimate in the protocol. {However, the monogamy relations derived in the previous section suggest that it might be advantageous to express $H(A|B)$ in terms of Bob's success probability, which potentially allows us to bound Eve's success. In order to be able to express $H(A|B)$ in terms of Bob's success probability,}
we assume that independently of the settings $x$ and $y$, Bob's output is $b = x_y$ with probability $S^\textit{rac}_B$ and $b \neq x_y$ with probability $1 - S^\textit{rac}_B$ (this holds in the case of the optimal quantum success probability, or when one assumes a noisy scenario with uniform noise, e.g.~employing the noisy states $\tilde{\rho}_x = \eta \rho_x + (1-\eta) \frac{\mathbb{I}}{d}$, where $\rho_x$ are the optimal states and $d$ is the dimension). Under these assumptions, $H(A|B) = h(S^\textit{rac}_B)$, where $h$ is the binary entropy, and we obtain a bound on the key rate in terms of the success probabilities of Eve and Bob,
\begin{equation}
    r \ge -\log S^\textit{rac}_E - h(S^\textit{rac}_B).
\end{equation}

{Since the logarithm is monotonically increasing,} in order to obtain a lower bound on the key rate in terms of Bob's success probability, $S^\textit{rac}_B$, it is enough to upper bound $S^\textit{rac}_E$ in terms of $S^\textit{rac}_B${. This} can be done using our semidefinite relaxation techniques, analogously to the monogamy relations obtained in Section~\ref{sec:monogamy}. We employed the projection-based relaxation $Q^{\Pi}_{1 + BE}$, and found that for a nontrivial region, $S^\textit{rac}_B\gtrapprox0.776$, below the optimal success probability $S^\textit{rac}_Q\approx0.7886751$, it is possible to extract a secure key from the observed behaviour (see Fig.~\ref{keyrate}). Therefore, we have demonstrated that preparation contextuality leads to secure semi-device-independent quantum cryptographic protocols\footnote{We note that in Ref.~\cite{KCBS_QKD} the authors study quantum key distribution based on the monogamy of the  Klyachko--Can--Binicioglu--Shumovsky-type contextuality, and in Ref.~\cite{QKD_mmt} the authors propose a quantum key distribution protocol based on generalised measurement contextuality. However, up to our knowledge, our proposed protocol is the first demonstration of the potential of preparation contextuality for semi-device-independent quantum key distribution.}.

\begin{figure}
    \centering
    \includegraphics[width=\linewidth]{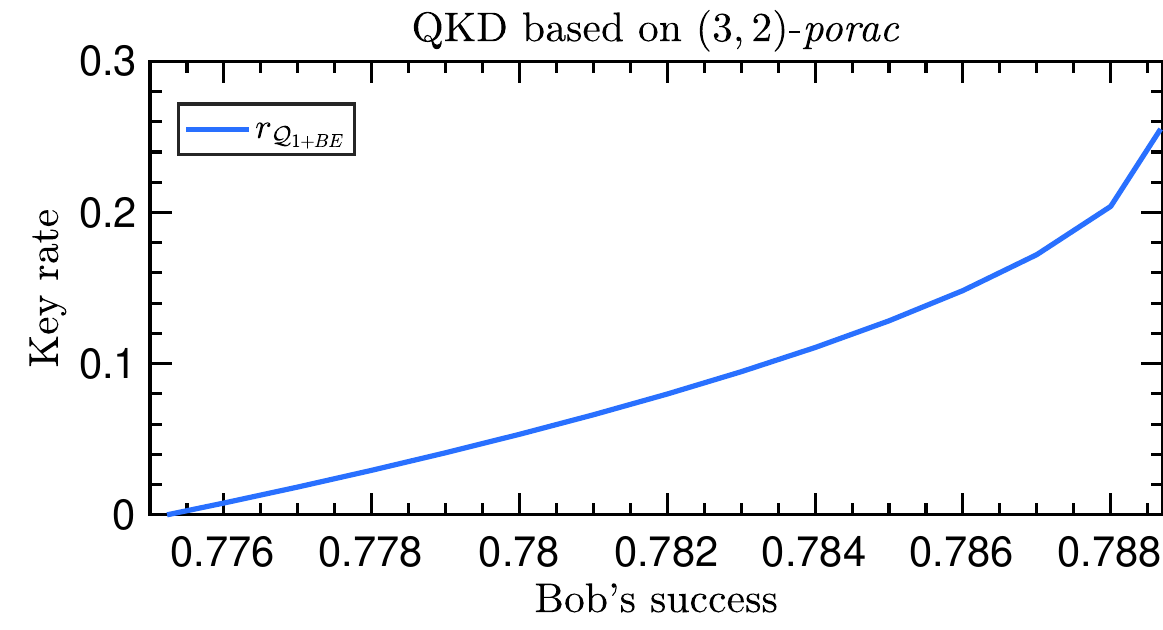}
    \caption{Lower bound on the key rate as a function of Bob's success probability employing the $(3,2)$-\textit{porac} based key distribution scheme described in the main text, computed via the relaxation $\cQ^{\Pi}_{1+BE}$. We see that for success probabilities $S^\textit{rac}_B\gtrapprox 0.776$ close to the optimal one, Alice and Bob are able to extract a secure key. If the parties find the success metric to be lower than this threshold value, the parties abort the protocol in light of the potential presence of an eavesdropper. \label{keyrate}}
\end{figure}

%% file: Summary.tex
\section{Summary}
 In this work, we sought to characterise the behaviours appearing in generalised contextuality scenarios with operational equivalences of preparations and measurements. To this end, we {laid} down a framework {for} rigorously defining contextuality scenarios {and the associated} polytope of valid contextual behaviours, $\mathcal{C}$, in prepare-and-measure experiments. We studied the set of quantum contextual behaviours, $\mathcal{Q}$. Moreover, we defined subsets of quantum contextual behaviours, namely, $\mathcal{Q}^{\Pi}$ and $\mathcal{Q}^{\Psi}$, wherein the {experimental} preparation and measurement procedures are restricted to pure states, and projective measurements, respectively. We established a number of general properties of these sets of contextual behaviours, namely, the convexity of $\mathcal{Q}$ (Observation \ref{QisConvex}), the convexity of $\mathcal{Q}^\Pi$ (Observation \ref{QPiisConvex}), and the absolute absence of contextuality in scenarios with only trivial operational equivalences of preparations (Observation \ref{NCequalsCBezPrepEq}) and in prepare-and-measure experiments with three or fewer preparations (Observation \ref{NCequalsCLessThanThreePreps}). Moreover, in contextuality scenarios with only trivial operational equivalences of measurement effects, it is enough to consider projective measurements (Observation \ref{QEqualsQPi}). We also find that in some scenarios the preparations and measurements actually appearing in an experiment (as opposed to their hypothetical counterparts appearing only in operational equivalences) must be mixed states and unsharp measurements, respectively, in order to violate noncontextuality, or to produce maximally contextual quantum behaviours (Observations \ref{QNotEqualsQPsi} and \ref{QNotEqualsQPi}). 
    
 With the framework in place, we focused on developing semidefinite programming techniques to bound the amount of quantum contextuality, as measured by the maximum quantum violation of noncontextuality inequalities. However, the fact that one needs to consider mixed states and unsharp measurements presents new challenges. In particular, one cannot employ simple modifications of the known semidefinite programming techniques such as that of Ref.~\cite{navascues2008convergent}. To address these issues, we formulate novel semidefinite programming relaxations. First, to address the inadequacy of pure states in contextuality scenarios, our relaxation employs multiple moment matrices hinged on the corresponding quantum preparations. The operational equivalences of preparations then semantically extend to linear constraints on the moment matrices themselves. These constraints are notably independent of the particulars of the operators indexing the moment matrices. Second, to address the inadequacy of projective measurements, we devise a novel formulation of semidefinite programming relaxations, entailing moment matrices indexed exclusively by monomials of unitary operators. This unitary operator parameterisation may be of independent interest to readers interested in noncommuting polynomial optimisation methods. Our formulation allows one to retrieve upper bounds on the quantum violation of noncontextuality inequalities with arbitrary operational equivalences. 
 
 The resulting hierarchies of semidefinite programs are efficient, as they provide tight bounds already at low levels. Consequently, we retrieve upper bounds on maximum quantum values, along with numerical proofs of their tightness, for more than twenty-one distinct noncontextuality inequalities, spanning a diverse selection of contextuality scenarios. To demonstrate the ease of implementation of these techniques, we provide a code tutorial (Appendix \ref{app:tutorial}) for the avid reader. 
 
 Moreover, these techniques also allow us to bound functionals that are relevant for applications, specifically, to derive monogamy relations for preparation contextuality, and to lower bound secure key rates of a novel semi-device-independent quantum key distribution scheme fuelled by quantum advantage in {the} $(3,2)$-\textit{porac} task and the monogamy of preparation contextuality.
 
{ While in this work we have considered contextuality scenarios with operational equivalences of preparations and measurements, in general, contextuality scenarios may also include operational equivalences of transformations \cite{spekkens2005contextuality}. {E}xtending the framework and the SDP technique so as to include operational equivalences of transformations forms an essential future research avenue. {Additionally, proving (or disproving) the convergence of our SDP hierarchy to the set of quantum contextual behaviours (in analogy with the convergence of the Navascu\'es-Pironio-Ac\'in hierarchy, approximating the quantum set of behaviours in Bell scenarios~\cite{navascues2008convergent}) presents another immediate research direction.}}

%% file: Appendix.tex
\clearpage
\onecolumngrid
\appendix

\section{Proofs}\label{app:proofs}

In this section we recall the statements from the main text and provide their proofs.

\setcounter{prop}{0}

\setcounter{lem}{0}
\begin{prop} \label{QisConvex_app}
For all contextuality scenarios the set of quantum behaviours $\mathcal{Q}$ is convex.
\end{prop}
\begin{proof}
Let us consider two quantum behaviours appearing in an arbitrary contextuality scenario, $\pp, \tilde{\pp} \in \cQ$, such that $p(k|x,y) = \tr( \rho_x M^y_k )$ and $\tilde{p}(k|x,y) = \tr( \tilde{\rho}_x \tilde{M}^y_k )$, where $\rho_x$ and $M^y_k$ are quantum states and measurements on a Hilbert space $\cH$, and $\tilde{\rho}_x$ and $\tilde{M}^y_k$ are quantum states and measurements on a Hilbert space $\tilde{\cH}$, respectively. Furthermore, the quantum preparations $\rho_x$ and $\tilde{\rho_x}$ and the quantum measurement operators $M^y_k$ and $\tilde{M}^y_k$ satisfy operational equivalences of form \eqref{OEPQ}.

Consider an arbitrary convex combination, $\hat{\pp} = \lambda \pp + (1-\lambda) \tilde{\pp}$, where $\lambda \in [0,1]$, i.e.,
\begin{equation}
    \hat{p}(k|x,y) = \lambda p(k|xy) + (1-\lambda)\tilde{p}(k|x,y) ~~ \forall x,y,k.
\end{equation}
We show that $\hat{\pp} \in \cQ$ by providing explicit states and measurements, such that $\hat{p}(k|x,y) = \tr( \hat{\rho}_x \hat{M}^y_k)$. In particular, consider the states,
\begin{equation}
    \hat{\rho}_x = \lambda \rho_x \oplus (1-\lambda) \tilde{\rho}_x \in \cB_+(\cH \oplus \tilde{\cH}),
\end{equation}
and the measurement operators,
\begin{equation}
    \hat{M}^y_k = M^y_k \oplus \tilde{M}^y_k \in \cB_+(\cH \oplus \tilde{\cH}).
\end{equation}
It is straightforward to verify that these direct sums constitute valid quantum states and measurements on $\hat{\cH}=\cH \oplus \tilde{\cH}$, satisfy the same operational equivalence constraints \eqref{OEPQ}, and that they indeed reproduce the desired behaviour, $\hat{p}(k|xy) = \tr( \hat{\rho}_x \hat{M}^y_k)$. 

\end{proof}
\begin{prop} \label{QPiisConvex_app}
For all contextuality scenarios the subset of quantum behaviours $\mathcal{Q}^{\Pi}$ is convex.
\end{prop}
\begin{proof}
The proof is identical to the proof of the previous observation. Additionally, observe that if the measurement {operators} are {projections}, i.e., $(M^y_k)^2 = M^y_k$ and $(\tilde{M}^y_k)^2 = \tilde{M}^y_k$ for all $y$ and $k$, then the {the measurement operators of the direct sum measurements are also projections}, i.e., $(\hat{M}^y_k)^2 = \hat{M}^y_k$ for all $y$ and $k$. 
\end{proof}
\begin{prop} \label{NCequalsCBezPrepEq_app}
For all contextuality scenarios with only trivial equivalences of preparations, $T = (X,Y,{K},\emptyset,\mathcal{OE}_M)$, all behaviours are noncontextual, i.e., $\mathcal{NC}=\mathcal{Q}=\mathcal{C}$. 
\end{prop}
\begin{proof}
Let $\mathbf{p}^{\cC}\in\mathcal{C}$ {in scenario $T$} be a generic contextual behaviour with elements $p^\cC(k|x,y)$. Firstly, consider the quantum strategy in which $\rho_x=\ketbra{x}$ where $\{\ket{x}\}_{x\in[X]}$ is an orthonormal basis of $\mathbb{C}^X$ and $M^y_k=\sum_{x\in[X]}p{^\cC}(k|x,y)\ketbra{x}$. Since each $0\leq p{^\cC}(k|x,y)\leq 1$ we have that $0\leq M^y_k\leq \mathbb{I}$, and since $\sum_{k\in[K]}p^\cC(k|x,y)=1$ we find $\sum_{k\in[K]}M^y_k=\mathbb{I}$ for all $x\in[X]$ and $y\in[Y]$. Moreover, it is straightforward to verify that this strategy reproduces the given behaviour. Finally, noting that since $\mathbf{p}^\cC\in\cC$ and therefore satisfies Eq.~\eqref{OEPO}, we find that these measurements satisfy operational equivalence conditions {\eqref{OEPQ}} as,
{\begin{equation}\sum_{k\in[K],y\in [Y]} \beta^w_j(k,y)M^y_k =  \sum_{x\in[X]} \left(\sum_{k\in[K],y\in [Y]}\beta^w_j(k,y) p{^\cC}(k|x,y)\right)\ketbra{x}=\sum_{x\in[X]}e_w\ketbra{x}\equiv F_w,\end{equation} 
for all $w\in[W]$ and all $j\in[W_w]$.} 
{Thus we have shown} $\mathcal{Q}=\mathcal{C}$.

{Similarly, we may show $\mathcal{NC}=\cC$ as follows. Consider an ontological model with a discrete ontic state space $\Lambda\equiv[X]$, and let $\mu_x$ for $x\in[X]$ be epistemic states such that $\mu_x(x) = 1$, for all $x \in [X]$. Furthermore, consider the response schemes $\xi_{M_y}(k|x)=p^\cC(k|x,y)$, for all $x\in[X]=\Lambda$, $k\in[K]$ and $y\in [Y]$. Clearly, this strategy reproduces the observed behaviour as $\sum_{\lambda\in[X]} \mu_x(\lambda)\xi_{M_y}(k|\lambda)=p^\cC(k|x,y)$. Finally, we find this strategy satisfies the measurement equivalences in \eqref{OEPOEMMOntic} since 
\begin{equation}
\sum_{k\in[K],y\in [Y]} \beta^w_j(k,y) \xi_{M_y}(k|x)  = \sum_{k\in[K],y\in [Y]} \beta^w_j(k,y) p^\cC(k|x,y)=e_w, 
\end{equation} 
for all $x\in[X]$, $w\in[W]$ and all $j\in[W_w]$.}
\end{proof}
\begin{prop} \label{NCequalsCLessThanThreePreps_app}
For all contextuality scenarios with three or fewer preparation procedures, $T = (X \le 3,Y,{K},\mathcal{OE}_P,\mathcal{OE}_M)$, all behaviours are noncontextual, i.e., $\mathcal{NC}=\mathcal{Q}=\mathcal{C}$. 
\end{prop}

\begin{proof}
Notice that in such contextuality scenarios, up to relabelling, there is only one non-trivial family of sets of operational equivalences of preparations when $X=3$, namely, 
\begin{equation} \label{OEMThree}
P_2\simeq \alpha P_0+(1-\alpha)P_1
\end{equation}
for some $\alpha\in [0,1]$. {We will show the result directly, however note that it could also be shown using Observation \ref{NCequalsCBezPrepEq} and showing that since one of the preparations is a convex combination of the other two, the scenario reduces to that with two preparations and no preparation equivalences.}

Given a behaviour $\boldsymbol{p{^\cC}}\in\mathcal{C}$, one can verify that the quantum preparations represented by the density matrices $\rho_0=\ketbra{0}$, $\rho_1=\ketbra{1}$ and $\rho_2=\alpha\rho_0+(1-\alpha)\rho_1$ and the operators $M^y_k=p{^\cC}(k|0,y)\ketbra{0}+p{^\cC}(k|1,y)\ketbra{1}$ satisfying $0\leq M^y_k \leq \mathbb{I}$, and $\sum_{k\in[K]}M^y_k=\mathbb{I}$, reproduces the observed behaviour as $\tr(\rho_xM^y_k)=p{^\cC}(k|x,y)$. Moreover, these quantum preparations satisfy operational equivalences \eqref{OEMThree}. Hence, we have shown $\mathcal{Q}=\mathcal{C}$. 

Now consider an ontological model, {with} a discrete ontic state-space ${\Lambda}\equiv{ \{0, 1\} }$, {three} epistemic states {$ \mu_0 (0) = 1$, $\mu_1(0) = 0$, $\mu_2(0) = \alpha$}, and {$Y$ measurements, $M_y$, with }response schemes {$\xi_{M_y}(k|x)=p^\cC(k|x,y)$, for all $x\in\Lambda$ and $k\in[K]$}. This model not only reproduces the observed behaviour {as $\sum_{\lambda\in\{0,1\}}\mu_x(\lambda)\xi_{M_y}(k|\lambda)=p^\cC(k|x,y)$}, but the epistemic states also satisfy the {operational equivalences implied by \eqref{OEMThree}}. This implies that $\mathcal{NC}=\mathcal{C}$, which concludes the proof.
\end{proof}

\begin{prop} \label{QEqualsQPi_app}
For all contextuality scenarios with only trivial equivalences of measurement effects $T\equiv (X,Y,{K},\mathcal{OE}_P,\emptyset )$, projective measurements are sufficient to recover the quantum set, i.e., $\mathcal{Q}=\mathcal{Q}^\Pi$.
\end{prop}

\begin{proof}
For every quantum behaviour, $\mathbf{p} \in \mathcal{Q}$, appearing in an arbitrary contextuality scenario of the form $(X,Y,{K},\mathcal{OE}_P,\emptyset)$, there exist density operators $\{ \rho_x \}_{x=0}^{X-1}$ satisfying the operational constraints $\mathcal{OE}_P$ of the form \eqref{OEPQ}, and POVMs $\big\{ \{ M^y_k \}_{k = 0}^{K-1} \big\}_{y=0}^{Y-1}$, on a finite-dimensional Hilbert space $\mathcal{H}$, such that
\begin{equation}
    p(k|x,y) = \tr( \rho_x M^y_k )
\end{equation}
for all $x$, $y$, $k$. Starting from this construction, we show that there also exists a set of density operators $\{ \hat{\rho}_x \}_{x=0}^{X-1}$ satisfying the same operational constraints $\mathcal{OE}_P$, and \textit{projective} measurements $\big\{ \{ P^y_k \}_{k = 0}^{K-1} \big\}_{y=0}^{Y-1}$, on a finite-dimensional Hilbert space $\mathcal{K}$, such that
\begin{equation}
    p(k|x,y) = \tr( \hat{\rho}_x P^y_k )
\end{equation}
for all $x$, $y$, $k$.

Our main technical tool is Naimark's dilation theorem.
\begin{thm}[Naimark dilation theorem]
For every POVM $\{ M_k \}_{k = 0}^{K-1}$ on a finite-dimensional Hilbert space $\mathcal{H}$, there exists another finite-dimensional Hilbert space $\mathcal{K}$, an isometry $V: \mathcal{H} \to \mathcal{K}$ and a projective measurement $\{ P_k \}_{k = 0}^{K-1}$ on $\mathcal{K}$, such that for every state $\rho$ on $\mathcal{H}$, we have that
\begin{equation}
    \tr( \rho M_k ) = \tr( V \rho V^\dagger P_k ).
\end{equation}
\end{thm}

Let us apply Naimark's dilation theorem to the first measurement, $\{ M^0_k \}_{k = 0}^{K-1}$. In particular, there exists a finite-dimensional Hilbert space $\mathcal{K}_0$, an isometry $V_0: \mathcal{H} \to \mathcal{K}_0$, and a projective measurement $\{ P^0_k \}_{k = 0}^{K-1}$ on $\mathcal{K}_0$, such that
\begin{equation}
    \tr( \rho_x M^0_k ) = \tr( V_0 \rho_x V_0^\dagger P^0_k)
\end{equation}
for all $x$, $k$. Following from the isometry property, $V_0^\dagger V_0 = \mathbb{I}_\mathcal{H}$, it follows that for the rest of the measurement operators we have
\begin{equation}
    \tr( \rho_x M^y_k ) =
    \tr( V_0 \rho_x V_0^\dagger V_0 M^y_k V_0^\dagger ) 
\end{equation}
for all $x$ and $k$, and for all $y = 1, 2, \ldots, Y-1$.

Let us now define $\rho_x^{(0)} \equiv V_0 \rho_x V_0^\dagger$. It is clear that $\rho_x^{(0)} \ge 0$ and $\tr( \rho_x^{(0)} ) = 1$, and therefore it is a valid state on $\mathcal{K}_0$. Note that in order to define new measurement operators, it is not enough to consider the operators $V_0 M^y_k V_0^\dagger$, since the sum of these over $k$ is $V_0V_0^\dagger$, which is not equal to the identity in general. Let us therefore define $M^{y,(0)}_k \equiv V_0 M^y_k V_0^\dagger + \mathbb{I}_{\mathcal{K}_0} - V_0V_0^\dagger$ for all $y = 1, 2, \ldots, Y-1$ and $M^{0,(0)}_k \equiv P^0_k$. It is straightforward to verify that all $M^{y,(0)}_k \ge 0$ (since they are sums of positive semidefinite operators), and that $\sum_k M^{y,(0)}_k = \mathbb{I}_{\mathcal{K}_0}$ for all $y$. Moreover, it is also easy to see that
\begin{equation}
    p(k|x,y) = \tr( \rho_x^{(0)} M^{y,(0)}_k )
\end{equation}
for all $x$, $y$, $k$, and by construction, the measurement $\{ M^{0,(0)}_k \}_{k = 0}^{K-1}$ is projective.

In the next step, we apply Naimark's dilation theorem to the measurement $\{ M^{1,(0)}_k \}_{k = 0}^{K-1}$. In particular, there exists a finite-dimensional Hilbert space $\mathcal{K}_1$, an isometry $V_1 : \mathcal{K}_0 \to \mathcal{K}_1$ and a projective measurement $\{ P^1_k \}_{k = 1}^{K - 1}$ on $\mathcal{K}_1$, such that
\begin{equation}
    \tr( \rho_x^{(0)} M^{1,(0)}_k ) = \tr( V_1 \rho_x^{(0)} V_1^\dagger P^1_k )
\end{equation}
for all $x$, $k$. Let us now define the states $\rho_x^{(1)} \equiv V_1 \rho_x^{(0)} V_1^\dagger$ and the measurement operators $M^{1,(1)}_k \equiv P^1_K$ and $M^{y,(1)}_k \equiv V_1 M^{y,(0)}_k V_1^\dagger + \mathbb{I}_{\mathcal{K}_1} - V_1 V_1^\dagger$ for all $y \neq 1$, on the Hilbert space $\mathcal{K}_1$. It is straightforward to verify that that $\{M^{0,(1)}_k \}_{k=0}^{K-1}$ is still a projective measurement, and that with this notation we have that
\begin{equation}
    p(k|,x,y) = \tr( \rho_x^{(1)} M^{y,(1)}_k )
\end{equation}
for all $x$, $y$, $k$, where the measurements $\{ M^{0,(1)}_k \}_{k = 0}^{K-1}$ and $\{ M^{1,(1)}_k \}_{k = 0}^{K-1}$ are projective.

It is clear now that the Naimark dilation theorem can be applied successively until every measurement is projective. Formally, for every $n = 0, \ldots, Y - 1$, we recursively define the measurement operators
\begin{equation}
    M^{y,(n)}_k \equiv
    \begin{cases}
    P^n_k & \text{if } y = n \\
    V_n M^{y,(n-1)}_k V_n^\dagger + \mathbb{I}_{\mathcal{K}_n} - V_n V_n^\dagger & \text{if } y \neq n,
    \end{cases}
\end{equation}
acting on the Hilbert space $\mathcal{K}_n$, which, together with the projective measurement $\{ P^n_k \}_{k = 0}^{K-1}$, and the isometry $V_n: \mathcal{K}_{n-1} \to \mathcal{K}_n$, describes the Naimark dilation of the measurement $\{M^{n,(n-1)}_k\}_{k=0}^{K-1}$. In order to be able to start the recursive definition, we set $M^{y,(-1)}_k \equiv M^y_k$ for all $y$, $k$ and $\mathcal{K}_{-1} \equiv \mathcal{H}$. Similarly, we recursively define the states
\begin{equation}
    \rho_x^{(n)} \equiv V_n \rho_x^{(n-1)} V_n^\dagger
\end{equation}
on the Hilbert space $\mathcal{K}_n$, and set $\rho_x^{(-1)} \equiv \rho_x$. It follows that the measurements $\big\{ \{ M^{y,(Y-1)}_k \}_{k=0}^{K-1} \big\}_{y=0}^{Y-1}$ are all projective on the Hilbert space $\mathcal{K}_{Y-1}$, and the states $\rho_x^{(Y-1)}$, on the same Hilbert space, satisfy the operational constraints $\mathcal{OE}_P$, due to the linearity of the transformation $\rho_x \mapsto \rho_x^{(Y-1)} \equiv V_{Y-1} V_{Y-2} \cdots V_0 \rho_x V_0^\dagger \ldots V_{Y-2}^\dagger V_{Y-1}^\dagger$. Moreover, we have that
\begin{equation}
    p(k|x,y) = \tr( \rho_x^{(Y-1)} M^{y,(Y-1)}_k )
\end{equation}
for all $k$, $x$, and $y$, which concludes the proof.
\end{proof}

\begin{prop} \label{QNotEqualsQPsi_app}
In general, for contextuality scenarios with non-trivial equivalences of preparations, pure states cannot produce all quantum contextual behaviours, i.e., $\mathcal{Q}^\Psi \subsetneq \mathcal{Q}$.
\end{prop}
\begin{proof}
In the simplest non-trivial contextuality scenario $(4,2,2,\mathcal{OE}_P,\emptyset)$, consider the $(2,2)$-\textit\textit{} setup with an operational equivalence of the form $\frac{1}{2}P_{00}+\frac{1}{2}P_{11}\simeq P_{01}$. Now, if we demand that the involved preparations are pure quantum states, i.e.,{ $\rho^2_{x_0,x_1}=\rho_{x_0,x_1}$ for all $x_0,x_1\in\{0,1\}$,} it is straightforward to see that the constraints implied by the given operational equivalence can only be satisfied when the three quantum preparations featuring in the operational equivalence condition are identical. This observation in turn leads to the upper bound on the performance of pure states in the task, $S^{\textit{rac}}_{\cQ^{\Psi}}=\frac{3}{4}$, {saturated by} the strategy, $\rho_{1,0}=\ketbra{0}$ and $\rho_{1,1}=\rho_{0,0}=\ketbra{1}$ {and $M^0_0=M^1_1=\ketbra{1}$}. {One can see that this strategy is indeed optimal from} the upper bound retrieved from the first level of our pure-state based relaxation, $S^{\textit{rac}}_{\cQ^\Psi_1}=\frac{3}{4}${, or analytically}.
Now, the noncontextual bound, quantum lower bound, and the upper bound obtained via the first level of our projection-based and unitary-based hierarchy are identical $S^\textit{rac}_{\mathcal{NC}}=S^\textit{rac}_{\cQ_L}=S^\textit{rac}_{\cQ^{\Pi}_1}=\frac{7}{8}$, {showing that $S^\textit{rac}_{Q^\Psi}\lneq S^\textit{rac}_{Q}$ and hence $\mathcal{Q}^\Psi \subsetneq \mathcal{Q}$}.

 {Alternatively, we obtain a numerical proof by considering the scenario $(4,2,2,\mathcal{OE}_P(\alpha),\emptyset)$ and noting that there exist non-empty intervals of the value $\alpha$ for which $S^\textit{rac}_{\cQ^\Psi_1}<S^\textit{rac}_{Q}$ as shown in Fig. \ref{batman}}. 
\end{proof}
\begin{prop} \label{QNotEqualsQPi_app}
In general, for contextuality scenarios with non-trivial equivalences of measurement effects, projective measurements cannot produce all quantum contextual behaviours, i.e., $\mathcal{Q}^\Pi \subsetneq \mathcal{Q}$.
\end{prop}
\begin{proof}
{Consider} the contextuality scenario $(3^2,2,3,\mathcal{OE}_P,\mathcal{OE}_M)$ {motivated by the } $(2,3)$-\textit{mporac} {task}, a generalization of {the} $(2,2)$-\textit{mporac}, wherein Alice is encoding two trits $x_0,x_1 \in [3]$. Operational equivalences of preparations are such that the parity trit $(x_0+x_1) \mod{3}$ remains secret{, explicitly the hypothetical preparations $\frac13(P_{a_j}+P_{b_j}+P_{c_j})$ are equivalent for all $j\in\{0,1,2\}$ where $a_j,b_j,c_j$ are the three pairs of trits summing to $j$}. 
Additionally, Bob has two three outcome measurements {with operators} $\{M^y_k\}^{1,2}_{y=0,k=0}$ which satisfy the following equivalence conditions of measurement effects,
\begin{equation} \label{mporac23}
    \frac{1}{2}[0|M_0] + \frac{1}{2}[0|M_1] \simeq \frac{1}{2}[1|M_0] + \frac{1}{2}[1|M_1] \simeq \frac{1}{2}[2|M_0] + \frac{1}{2}[2|M_1]. 
\end{equation}
We shall now demonstrate that the operational equivalences \eqref{mporac23} do not allow for any projective measurements, rendering the set of quantum contextual behaviour obtained when the individual quantum measurements are projective $\mathbb{Q}^\Pi$ to be $\emptyset$. It is straigtforward to see that the operational equivalences \eqref{mporac23} necessitate the following constraints on quantum measurement effects, 
\begin{equation} \label{23mporac}
    \frac{1}{2}\left(
                \begin{array}{ll}
                  M^0_0+M^1_0\\
                  M^0_1+M^1_1\\
                  M^0_2+M^1_2
                \end{array}
    \right) \simeq
    \left(
                \begin{array}{ll}
                  \frac{\mathbb{I}}{3}\\
                  \frac{\mathbb{I}}{3}\\
                  \frac{\mathbb{I}}{3}
                \end{array}
    \right).
\end{equation}
Let us now assume that each $M^y_k$ is a projection, {i.e.,} $(M^y_k)^2=M^y_k$. Now from the equation \eqref{23mporac} we have $M^1_k = \frac{2\mathbb{I}}{3}-M^0_k$, for all $k \in{\{0,1,2\}}$, which leads us to,
\begin{eqnarray} \nonumber
& (M^1_k)^2 & =  \frac{4\mathbb{I}}{9}-\frac{4M^0_k}{3}+(M^0_k)^2, \\
& & = \frac{4\mathbb{I}}{9}-\frac{M^0_k}{3},
\end{eqnarray}
where the second equality follows from the projectivity assumption $(M^0_k)^2=M^0_k$. {On the other hand} the assumption $(M^1_k)^2=M^1_k$ implies   $\frac{2\mathbb{I}}{3}-M^0_k=\frac{4\mathbb{I}}{9}-\frac{M^0_k}{3}$ which yields $M^0_k=\frac{\mathbb{I}}{3}$ for any $k\in {\{0,1,2\}}$, which is not a projection, and therefore contradicts our initial projectivity assumption. {The maximal noncontextual success probability of $(2,3)$-\textit{mporac} is greater than $\frac{1}{3}$, i.e., $S_\mathcal{NC} = \frac{1}{2}$.} Finally, {there exist POVMs satisfying the constraints \eqref{23mporac}, such as those} depicted in Fig. \ref{2to13}{. The states and POVMs in Fig. \ref{2to13} show that $\cQ$ is non-empty} which concludes the proof (the strategy also surpass{es} the noncontextual bound $S_\mathcal{NC} = \frac{1}{2}$ and achieve a success probability of $S_{\mathcal{Q}_L} \approx 0.5257834$, while the first level of our unitary-based hierarchy returns the bound $S_{\mathcal{Q}_1}=0.5555555$). 

We stumbled upon another numerical proof for the desired thesis while testing our relaxations on the facet inequalities in the contextuality scenario
$(6,3,2,\mathcal{U}_P,\mathcal{U}_M)$ from \cite{PhysRevA.97.062103}. Specifically, the seventh inequality in Table \ref{SchidtSpekkens} (which was listed as a facet inequality in \cite{PhysRevA.97.062103} due to a typing error) has a lower value when we restrict ourselves to projective measurements $S_{\mathcal{Q}^\Pi_1}=S_{\mathcal{Q}^\Pi_L}=3.4641016$, as compared to even the noncontextual bound $S_\mathcal{NC}=\frac{7}{2}$, and to the quantum lower bound $S_{\mathcal{Q}_L}=\frac{7}{2}$, and the bounds obtained from the first and second level of our unitary based relaxation $S_{\mathcal{Q}_1}=3.5552760$ and $S_{\mathcal{Q}_2}=3.5$.

\begin{figure}
    \centering
    \includegraphics[width=0.5\textwidth]{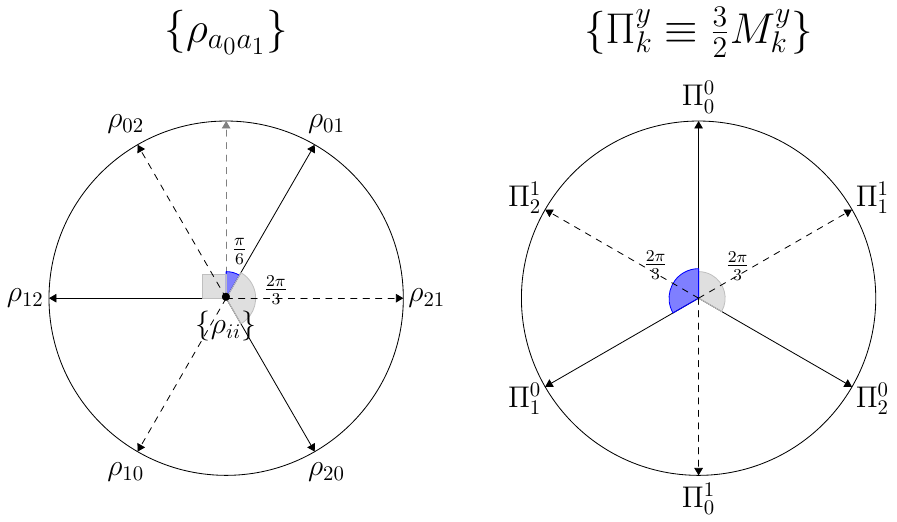}
    \caption{\label{2to13} A quantum strategy for $(2,3)$-\textit{mporac} task, in the contextuality scenario $(3^2,2,3,\mathcal{OE}_P,\mathcal{OE}_M)$, wherein the nine preparation settings are indexed by two trits $x_0,x_1 \in [3]$, the set $\mathcal{OE}_P$ entails operational equivalences of preparations such that the parity trit $(x_0+x_1)\mod 3$ remains a secret, and $\mathcal{OE}_M$ entails operational equivalences of measurement effects of the form \eqref{mporac23}. This strategy violates the noncontextual bound and attains $S^\textit{rac}_{\mathcal{Q}_L} = 0.5257834 > S^\textit{rac}_\mathcal{NC} = 0.5 > S^\textit{rac}_{\mathcal{Q}^\Pi} = \frac{1}{3}$. Here, the displayed {projections} (right) $\Pi^y_k$ are {scaled down} to {become measurement operators} $M^y_k=\frac{2}{3}\Pi^y_k$ for all $y\in[2]$ and $k\in[3]$, which allows for the violation of the noncontextual bound, and a success metric gain of $S^\textit{rac}_{\mathcal{Q}_L} \approx 0.5257834-0.3333333 =0.1924501$. This constitutes a part of the proof of Observation \ref{QNotEqualsQPi_app}, wherein we show that a any strategy wherein all effects are projections is simply not allowed in this task, thereby demonstrating that in certain universal contextuality scenarios unsharp measurements can outperform the sharp measurements.} 
\end{figure}
\end{proof}
\begin{lem} \label{unitarylemma_app} Any operator $0 \le M \le \mathbb{I}$ on a finite dimensional Hilbert space $\mathcal{H}$ can be written as $M=\frac{\mathbb{I}}{2}+\frac{U+U^\dagger}{4}$, where $U$ is a unitary operator on $\mathcal{H}$.
\end{lem}
\begin{proof}
Every such operator $M$ can be written in terms of its spectral decomposition,
\begin{equation}
    M = \sum_{j = 0}^{\dim \cH - 1} \lambda_j \ketbra{j}{j},
\end{equation}
where $\{ \ket{j} \}_{j=0}^{\dim \cH - 1}$ is an orthonormal basis on $\cH$ (the eigenbasis of $M$), and $0 \le \lambda_j \le 1$ for all $j$. Then, consider the unitary operator,
\begin{equation} \nonumber
    U = \sum_{j = 0}^{\dim \cH - 1} \mathrm{e}^{\mathrm{i} \alpha_j} \ketbra{j}{j},
\end{equation}
such that $\lambda_j = \frac12( 1 + \cos \alpha_j)$. Notice that such a set $\{ \alpha_j \}_j$ can always be found for any set $\{ \lambda_j \}_j$ satisfying $0 \le \lambda_j \le 1$ for all $j$. Therefore, it holds that
\begin{equation} \nonumber
    M = \sum_{j = 0}^{\dim \cH - 1} \frac12( 1 + \cos \alpha_j) \ketbra{j}{j} = \sum_{j = 0}^{\dim \cH - 1} \Big[ \frac12 + \frac14 (\mathrm{e}^{\mathrm{i} \alpha_j} + \mathrm{e}^{-\mathrm{i} \alpha_j}) \Big] \ketbra{j}{j} = \frac{\mathbb{I}}{2}+\frac{U+U^\dagger}{4}.
\end{equation}
\end{proof}
\begin{lem} \label{mporac}
The maximal noncontextual success probability of $(n,2)$-\textit{mporac} remains the same as the {maximal} noncontextual success probability of {the} $(n,2)$-\textit{porac}, {that is,} $S^\textit{rac}_{\mathcal{NC}}=\frac{1}{2}(1+\frac{1}{n})$.
\end{lem}
\begin{proof}
As the $(n,2)$-\textit{mporac} {has} additional {measurement equivalence constraints} compared to the standard $(n,2)$-\textit{porac}---which has the noncontextual bound $\frac{1}{2}(1+\frac{1}{n})$-- we have {that} $S^\textit{rac}_{\mathcal{NC}}\leq \frac{1}{2}(1+\frac{1}{n})$. All we need to prove is that $\frac{1}{2}(1+\frac{1}{n})$ is a viable success probability in presence of the additional operational equivalences of measurement effects. Let us consider a binary ontic state space $\Lambda={ \{0,1\} }$. {Further, we consider the epistemic states that simply encode} the first bit of Alice, {that is,} $\mu_x(x_0)=1$, and are consequently noncontextual. {Finally, let the response schemes be given by} $\xi{_{M_0}}(\lambda|\lambda)=1$, and $\xi{_{M_y}}(\lambda|\lambda)=\frac{n-2}{2(n-1)}$ {for all $y\neq0$}. These response schemes satisfy the {additional} operational equivalences {of the $(n,2)$-\textit{mporac} task}, specifically, $\frac{1}{n}\sum_{y\in[n]} \xi{_{M_y}}(0|{.}) = \frac{1}{n}\sum_{y\in[n]} \xi{_{M_y}}(1|.) = \frac{1}{2}$, and hence are noncontextual.
It is straightforward to verify that these epistemic states and response schemes achieve the desired success probability $S^{rac}_{\mathcal{NC}}=\frac{1}{2}(1+\frac{1}{n})$, which concludes the proof.
\end{proof}
\clearpage
\input{tute}

\clearpage
\section{Optimal states and measurements}
\begin{figure}[h]
    \centering
    \includegraphics[width=0.63\linewidth]{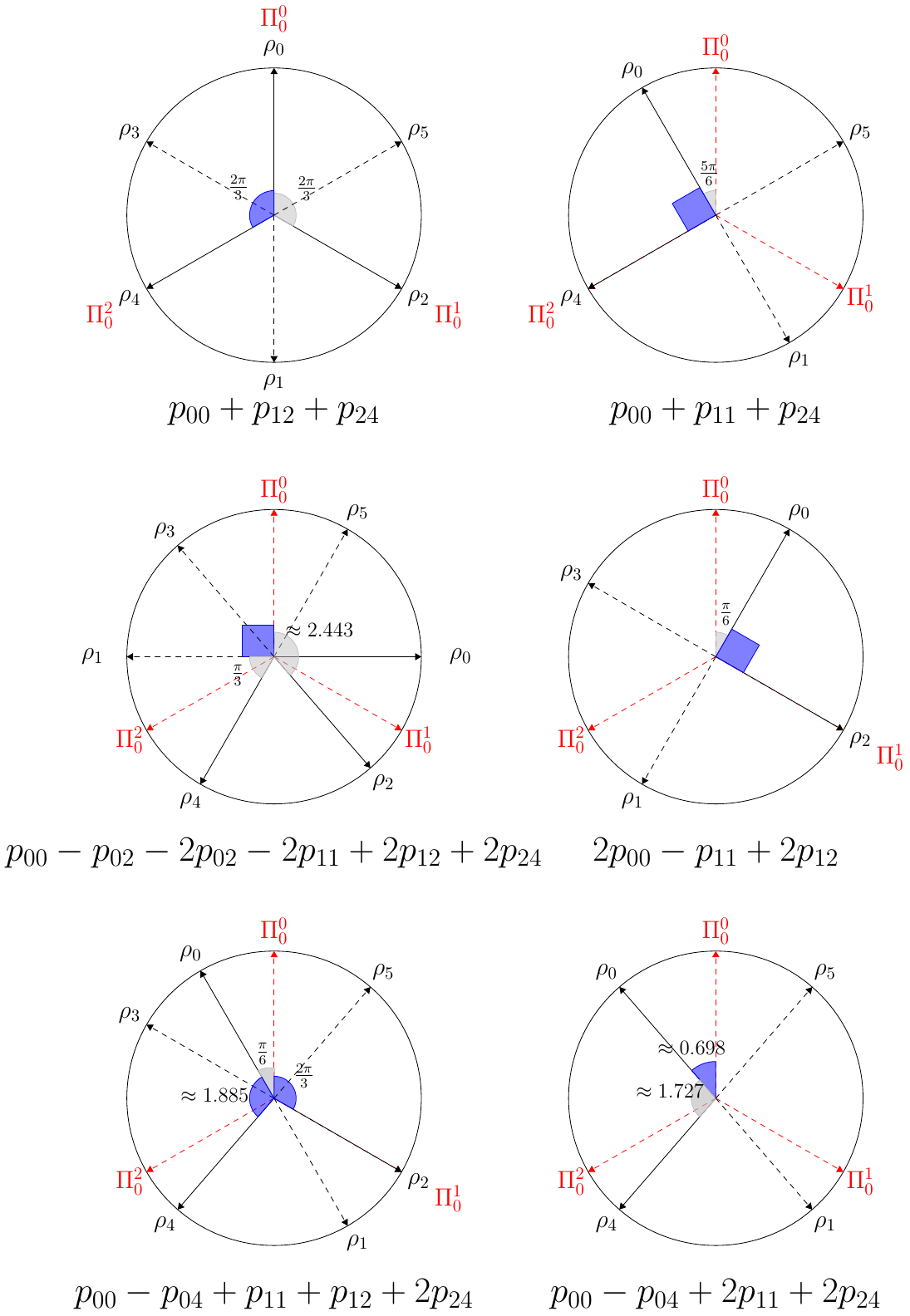}
    \caption{\label{StatesAndMeas} Quantum protocols that attain maximal quantum values (upto machine precision) for all six facet inequalities in the contextuality scenario $(6,3,2,\mathcal{U}_P,\mathcal{U}_M)$, where $p_{yx}=p(0|x,y)$. The states and the measurements are depicted on a plane of the Bloch sphere. In each case, the states $\rho_0,\rho_2,\rho_4$ are depicted with solid black arrows, the states $\rho_1,\rho_3,\rho_5$ are represented by dashed black arrows, and the projections $\Pi^0_0,\Pi^1_0,\Pi^2_0$ are represented by dashed red arrows. In each case, $\rho_0+\rho_1=\rho_2+\rho_3=\rho_4+\rho_5=\mathbb{I}$, {according to} $\mathcal{U}_P$. Observe that for the second, fourth and sixth inequalities, we have skipped pairs of states that do not feature in the expression of the respective facet inequalities, as these can be arbitrary, as long as they add up to $\mathbb{I}$. Finally, the measurements remain unaltered for all six facet inequalities, {that is,} $\Pi^0_0+\Pi^1_0+\Pi^2_0 = \frac{3\mathbb{I}}{2}$, {according to} $\mathcal{U}_M$.}
\end{figure}

%% file: tute.tex
\section{Tutorial code}\label{app:tutorial}
 As a demonstration of the ease of implementation of the relaxations presented in this work, we provide a basic verbose code for the universal contextuality scenario $(6,3,2,\mathcal{U}_P,\mathcal{U}_M)$, where we have the following operational equivalences,
\begin{align}
\label{OEP632}
 \frac{1}{2}(P_{0}+P_{1}) \simeq \frac{1}{2}(P_{2}&+P_{3})\simeq \frac{1}{2}(P_{4}+P_{5}),   \\ \label{OEM632} 
\frac{1}{3}([0|M_0]+[0|M_1]+[0|M_2]) &\simeq \frac{1}{3}([1|M_0]+[1|M_1]+[1|M_2]).
\end{align}
Here we will be employing \textsc{Matlab} in conjunction with a standard optimization toolbox \href{https://github.com/yalmip/YALMIP}{\textsc{YALMIP}}, and \href{https://github.com/sqlp/sdpt3}{\textsc{SDPT3}} as our SDP solver. 

For this scenario we will employ the operator sequence {$\mathcal{O}=(\mathbb{I})\mathbin\Vert (U^y_k,{U^y_k}^\dagger)^{Y-1,K-2}_{y=0,k=0}$}. 
We begin by declaring the specifications of the prepare-and-measure experiment, the total number of operators, empty arrays for the constraints, a probability cell for storing the observed behaviour, an empty array to for the results, and $X$ SDP variables, 
\begin{lstlisting}
X = 6; % six settings for the preparation device
Y = 3; % three settings for the measurement device
K = 2; % binary outcomes
O = 2*Y*(K-1)+1; % total number of operators in our list
conP = []; % an empty list for constraints on the moment matrix level
conM = []; % an empty list for constraints on the substrate level
Prob = cell(X,Y,K-1); % a cell for probabilities
S = []; % an empty list for upper bounds
G = cell(X,1); % cell for X moment matrices
\end{lstlisting}
It is a feature of our formulation that one can easily segregate constraints on abstract moment matrices, from those that depend specifically on the operators in the operator list. The former includes the following constraints that ensue from the positive semidefiniteness of the quantum states, and the equivalence conditions of the preparations \eqref{OEP632}, 
\begin{lstlisting}
for x = 0:X-1
    G{x+1}=sdpvar(O,O,'hermitian','complex'); % declaration of our SDP variables
    conP = [conP;G{x+1} >= 0]; % semi-definite constraints
end
conP = [conP; G{1} + G{2} == G{3} + G{4}; G{1} + G{2} == G{5} + G{6}]; % preparation equivalences 
\end{lstlisting}
Now, let us collect constraints that are specific to the operator list we are employing.
First, we create an indexing function which returns the position of the operator in each moment matrix, given the operator specifiers,
\begin{lstlisting}
idx = @(y, k, u) 2*(K-1)*y + 2*k + u + 2; % function to return the position of the operators 
\end{lstlisting} 
where $U^{y,u=0}_k=U^{y}_k$, and $U^{y,u=1}_k={U^{y}_k}^\dagger$.
Let us now specify how the operators $\{{U^y_k}^\dagger\}^{Y-1,K-2}_{y=0,k=0}$ are related to $\{{U^y_k}\}^{Y-1,K-2}_{y=0,k=0}$, along with the constraints that ensue from the unitarity of the operators in our list,

\begin{lstlisting}
for x = 0:X-1
    for y = 0:Y-1
        for k = 0:K-2 
            conM = [conM; G{x+1}(1,idx(y,k,0)) == G{x+1}(idx(y,k,1),1)]; 
            conM = [conM; G{x+1}(idx(y,k,0),1) == G{x+1}(1,idx(y,k,1))];
        end
    end
    for j = 1:O
        conM = [conM; G{x+1}(j,j) == 1]; % unitarity constraints
    end
end
\end{lstlisting}

The measurement equivalences \eqref{OEM632} amount to additional constraints on the unitary operators, specifically, 
\begin{equation}
    \sum^{Y-1}_{y=0}(U^y_0+{U^y_0}^\dagger) = 0
\end{equation}
which can be enforced in the following way, specific to our operator list,
\begin{lstlisting}
for x = 0:X-1
    for j = 1:O
        sum1 = 0; sum2 = 0; 
        for y = 0:Y-1
            for k = 0:K-2 
                sum1 = sum1 + G{x+1}(j,idx(y,k,0)) + G{x+1}(j,idx(y,k,1));
                sum2 = sum2 + G{x+1}(idx(y,k,0),j) + G{x+1}(idx(y,k,1),j);
            end
        end
        conM = [conM; sum1 == 0; sum2 == 0]; % measurement equivalences 
    end
end
\end{lstlisting}
We can now collect the observed data, $p(k|x,y) = \frac{1}{2}+\frac{1}{4}(\tr{\rho_x U^k_y} + \tr{\rho_x {U^k_y}^\dagger})$,
\begin{lstlisting}
for x = 0:X-1
    for y = 0:Y-1
        for k = 0:K-2
            Prob{x+1,y+1,k+1} = 0.5 + 0.25 * (G{x+1}(1, idx(y,k,0)) + G{x+1}(1, idx(y,k,1)));
        end
    end
end
\end{lstlisting}
Now we are  {in the position} to retrieve quantum bounds on facet inequalities in Table \ref{SchidtSpekkens} and in \cite{schmid2019characterization},
\begin{lstlisting}
S1 = real(Prob{1,1,1} + Prob{3,2,1} + Prob{5,3,1});
diagnostics = optimize([conP;conM], -S1, sdpsettings('solver', 'sdpt3'));
S = [S;value(S1)];
%-------------------------------------------------------------------------------  
S2 = real(Prob{1,1,1} + Prob{2,2,1} + Prob{5,3,1});
diagnostics = optimize([conP;conM], -S2, sdpsettings('solver', 'sdpt3'));
S = [S;value(S2)];
%------------------------------------------------------------------------------- 
S3 = real(Prob{1,1,1} - Prob{3,1,1} -2 * Prob{5,1,1} -2 * Prob{2,2,1} + 2 * Prob{3,2,1} + 2 * Prob{5,3,1});
diagnostics = optimize([conP;conM], -S3, sdpsettings('solver', 'sdpt3'));
S = [S;value(S3)];
%------------------------------------------------------------------------------- 
S4 = real(2* Prob{1,1,1} - Prob{2,2,1} +2* Prob{3,2,1}); 
diagnostics = optimize([conP;conM], -S4, sdpsettings('solver', 'sdpt3'));
S = [S;value(S4)];
%------------------------------------------------------------------------------- 
S5 = real(Prob{1,1,1} - Prob{5,1,1} +  Prob{2,2,1} + Prob{3,2,1} + 2 * Prob{5,3,1});
diagnostics = optimize([conP;conM], -S5, sdpsettings('solver', 'sdpt3'));
S = [S;value(S5)];
%------------------------------------------------------------------------------- 
S6 = real(Prob{1,1,1} - Prob{5,1,1} +  2*Prob{2,2,1} + 2 * Prob{5,3,1});
diagnostics = optimize([conP;conM], -S6, sdpsettings('solver', 'sdpt3'));
S = [S;value(S6)];
%------------------------------------------------------------------------------- 
S7 = real(Prob{1,1,1} - Prob{4,1,1} -2 * Prob{5,1,1} -2 * Prob{2,2,1} + 2 * Prob{3,2,1} + 2 * Prob{5,3,1});
diagnostics = optimize([conP;conM], -S7, sdpsettings('solver', 'sdpt3'));
S = [S;value(S7)];
\end{lstlisting}
where the first six expressions correspond to the facet inequalities of the noncontextual polytope in this scenario (see Table \ref{SchidtSpekkens}), while the seventh expression from equation $(26)$ in \cite{schmid2019characterization} is the third facet inequality with a typing error. {W}e can display the results by
\begin{lstlisting}
format long;
disp(S)
>>3.000000003034520
>>2.866025404556603
>>3.920951869054678
>>3.366025404977135
>>4.688901060747579
>>4.645751315302937
>>3.555276067392406
\end{lstlisting}